\documentclass[letterpaper,twocolumn,10pt]{article}
\usepackage{usenix2019_v3}

\usepackage{tikz}
\def\checkmark{\tikz\fill[scale=0.4](0,.35) -- (.25,0) -- (1,.7) -- (.25,.15) -- cycle;} 

\usepackage{amsmath}

\usepackage{caption}
\usepackage{enumitem}
\usepackage{csquotes}
\usepackage{todonotes}
\usepackage{algpseudocode}
\usepackage[linesnumbered,ruled,vlined]{algorithm2e}
\usepackage{tabularx}
\usepackage{graphicx} 
\usepackage{multirow}
\usepackage{booktabs}
\usepackage{rotating}
\usepackage{enumitem}
\usepackage{subcaption} 
\usepackage{amsmath}
\usepackage{amsthm}
\usepackage{amsmath}
\usepackage{amssymb}
\usepackage{amsfonts}
\usepackage{stmaryrd}
\usepackage{float}

\usepackage{array}
\newcolumntype{P}[1]{>{\centering\arraybackslash}p{#1}}
\newcolumntype{M}[1]{>{\centering\arraybackslash}m{#1}}

\usepackage{filecontents}

\usepackage{colortbl}
\usepackage{xcolor}

\newcommand{\secureDL}{\textbf{SecureDL}}

\newtheorem{assumption}{Assumption}
\newtheorem{theorem}{Theorem}

\begin{document}

\date{}


\title{\Large \bf Privacy-Preserving Aggregation for Decentralized Learning with Byzantine-Robustness}

\author{Ali Reza Ghavamipour, Benjamin Zi Hao Zhao, Oğuzhan Ersoy, Fatih Turkmen}

\maketitle

\begin{abstract}

Decentralized machine learning (DL) has been receiving an increasing interest recently due to the elimination of a single point of failure, present in Federated learning setting. 
Yet, it is threatened by the looming threat of Byzantine clients who intentionally disrupt the learning process by broadcasting arbitrary model updates to other clients, seeking to degrade the performance of the global model. 
In response, robust aggregation schemes have emerged as promising solutions to defend against such Byzantine clients, thereby enhancing the robustness of Decentralized Learning. 
Defenses against Byzantine adversaries, however, typically require access to the updates of other clients, a counterproductive privacy trade-off that in turn increases the risk of inference attacks on those same model updates.

In this paper, we introduce \secureDL, a novel DL protocol designed to enhance the security and privacy of DL against Byzantine threats. \secureDL~facilitates a collaborative defense, while protecting the privacy of clients' model updates through 
secure multiparty computation. 
The protocol employs efficient computation of cosine similarity and normalization of updates to robustly detect and exclude model updates detrimental to model convergence. 
By using MNIST, Fashion-MNIST, SVHN and CIFAR-10 datasets, we evaluated \secureDL~against various Byzantine attacks and compared its effectiveness with four existing defense mechanisms. 
Our experiments show that \secureDL~is effective even in the case of attacks by the malicious majority (e.g., 80\% Byzantine clients) while preserving high training accuracy.

\end{abstract}

\section{Introduction}



The traditional approach to training a machine learning model has been the aggregation of training dataset to a central server and performing the training locally on this server. However centralization is challenging when the data is geographically dispersed across fleets of devices or workers, which is further coupled with the increasing concerns around data privacy. In response, researchers have proposed alternatives to centralized learning as a means to train machine learning models locally (i.e., data locality) while preserving privacy~\cite{lian2017can}. 

Collaborative machine learning is a paradigm shift enabling diverse entities to collaboratively train models with improved accuracy and robustness, while maintaining privacy of collaborators. This approach is particularly beneficial when data centralization is impractical due to privacy, bandwidth, or regulatory constraints. Arguably, the most well-known approach of collaborative machine learning is Federated Learning (FL)~\cite{mcmahan2017communication}. FL features a central server which aggregates the individual model updates obtained from collaborators who train their models locally using their own private dataset. However, FL places an inherent trust in the central server, a potential single point of failure~\cite{lian2017can,bonawitz2019towards,ghavamipour2023federated}, that may result with the overlooking of potential issues related to security 
and computational efficiency.

Decentralized Learning (DL) is another approach that has gained popularity as a scalable and communication-efficient alternative to existing collaborative learning paradigms. Compared to FL schemes~\cite{lian2017can}, DL eliminates the need for having a central aggregator which in turn enhances client's privacy and prevents issues related to single point of failure~\cite{cheng2019towards, vogels2021relaysum, hu2019decentralized}. 
We illustrate the structural differences between FL and DL in Figure~\ref{DVvsFL}. Specifically, DL implements the global model training by aggregating the model parameters on-device with peer-to-peer message exchanges between clients (similar to gossip-based algorithms over random networks~\cite{hegedHus2019gossip}), instead of centralized aggregation. In DL, every client actively updates its own model by using local data and updates received from its neighbors that results with optimized learning in terms of model accuracy and convergence speed.

\begin{figure}[t]
\centering
  \includegraphics[width=\linewidth]{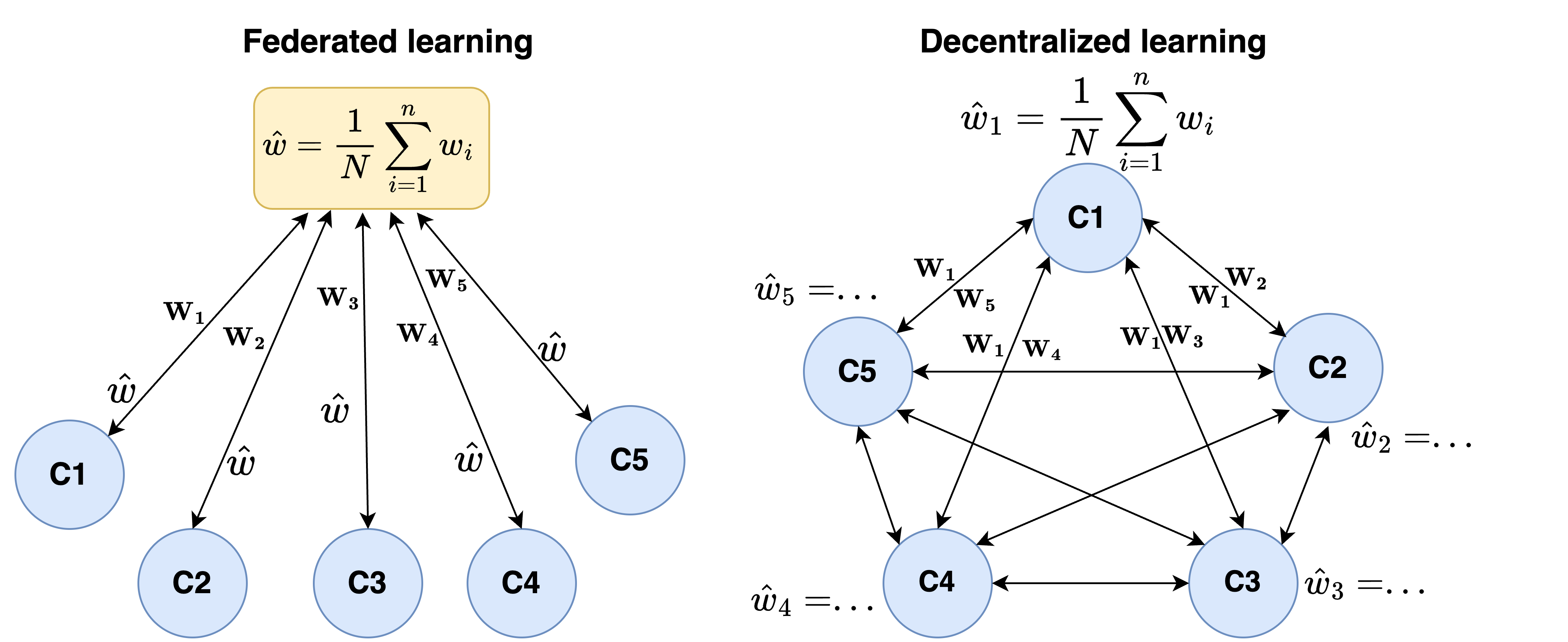}
  \caption{Collaborative learning system in a Federated (left) and decentralized fashion.}
  \label{DVvsFL}
\end{figure}

Given all these benefits, DL systems are particularly more susceptible to privacy attacks, both under passive and active security models, with distinct privacy challenges compared to those encountered in FL \cite{pasquini2023security}. More specifically, the architecture of DL broadens the attack surface, allowing any client within the system to perform privacy attacks, or result with more data leakages. 

In both FL and DL, the collaborative nature of these systems exposes them to significant risks from Byzantine adversaries. These adversaries can exploit the system by compromising genuine clients, thereby submitting updates that can undermine the model's integrity. Specifically, they can launch poisoning attacks in the form of data poisoning—tainting the local training data~\cite{biggio2012poisoning}—or model poisoning~\cite{fang2020local,bhagoji2019analyzing}—altering the updates sent to the central model. Such manipulations can drastically affect the global model's accuracy, leading to large number of incorrect predictions. This vulnerability highlights the pressing need for enhanced security protocols to protect the collaborative learning process from these adversarial threats.

\begin{table}[ht]
\centering
\scriptsize
\caption{Comparison of \secureDL~ with the existing robust aggregation schemes and their main defense mechanisms. Abbreviations: \textbf{EucD}: Euclidean distance, \textbf{PerC}: Performance Check, \textbf{CosS}: Cosine Similarity, \textbf{Coord}: Coordinate-wise operations, \textbf{Norm}: Normalization.}
\label{tab:existing}

\setlength{\tabcolsep}{6pt} 
\resizebox{1.0\columnwidth}{!}{%
\begin{tabular}{|p{0.9cm}|ccclcc|c|c|}
\hline
\multirow{3}{*}{Rules} &
  \multicolumn{6}{c|}{Main defense mechanisms} &
  \multirow{2}{*}{\begin{tabular}[c]{@{}c@{}}\# Byzantine \\ tolerance\end{tabular}} &
  \multirow{2}{*}{\begin{tabular}[c]{@{}c@{}}Privacy\end{tabular}} \\ \cline{2-7}
 &
  \multicolumn{1}{c|}{EucD} &
  \multicolumn{1}{c|}{PerC} &
  \multicolumn{2}{c|}{CosS} &
  \multicolumn{1}{c|}{Coord} &
  Norm &
   &
   \\
   \hline
DKrum &
  \multicolumn{1}{c|}{\checkmark} &
  \multicolumn{1}{c|}{} &
  \multicolumn{2}{c|}{} &
  \multicolumn{1}{c|}{} &
   &
  \begin{tabular}[c]{@{}c@{}}$B < {n}/{2}$\end{tabular} &
  No \\ \hline
BRIDGE &
  \multicolumn{1}{c|}{} &
  \multicolumn{1}{c|}{} &
  \multicolumn{2}{c|}{} &
  \multicolumn{1}{c|}{\checkmark} &
   &
  \begin{tabular}[c]{@{}c@{}}$B < {n}/{2}$\end{tabular} &
  No \\ \hline
MOZI &
  \multicolumn{1}{c|}{\checkmark} &
  \multicolumn{1}{c|}{\checkmark} &
  \multicolumn{2}{c|}{} &
  \multicolumn{1}{c|}{} &
   &
  $B \leq n-1$ &
  No \\ \hline
DMedian &
  \multicolumn{1}{c|}{} &
  \multicolumn{1}{c|}{} &
  \multicolumn{2}{c|}{} &
  \multicolumn{1}{c|}{\checkmark} &
   &
  \begin{tabular}[c]{@{}c@{}}$B < {n}/{2}$
\end{tabular} &
  No \\ \hline
\textbf{\secureDL} &
  \multicolumn{1}{c|}{} &
  \multicolumn{1}{c|}{} &
  \multicolumn{2}{c|}{\checkmark} &
  \multicolumn{1}{c|}{} &
  \checkmark &
  $B \leq n-1$ &
  Yes \\ \hline
\end{tabular}
}
\end{table}



Numerous techniques have been developed to prevent such poisoning attacks (mainly for FL) including coordinate-wise median~\cite{yin2018byzantine}, geometric median~\cite{chen2017distributed}, trimmed-mean~\cite{xie2018phocas}, and Krum~\cite{blanchard2017machine} as part of the defenses against both passive and active attackers, that are also applicable to the decentralized learning context~\cite{fang2019bridge,yang2019byrdie,peng2021byzantine,he2022byzantine}. 
Despite the progress, 
a limitation of these methods is their reliance on clients receiving model updates directly from their collaborators, and thus observing the models of (or be observed by) other clients which may compromise privacy if certain security measures are not in place. 
This consequently breaches a core tenant of DL\cite{pasquini2023security} that aims to provide privacy protection for each collaborator's model without the need of an aggregator.

In this paper, we present \secureDL, a privacy-preserving protocol to train machine learning models in a decentralized manner.
To the best of our knowledge, \secureDL~is the first such protocol that is also resilient to Byzantine attacks. 
By employing a robust defense mechanism against Byzantine clients based on Cosine-similarity and normalization that are coupled with a secret sharing mechanism, \secureDL~allows clients to collaboratively identify malicious updates while maintaining the privacy of each clients’ model updates.

The core of \secureDL~is the Byzantine-resilient aggregation mechanism, where each received update is scrutinised prior to aggregation, specifically, with careful characterization of the update direction. 
Additionally, \secureDL~ensures the coherence of updates from each client by calibrating their magnitude, aligning them within the same norm. 
To achieve these functionalities, we use a mix of secure comparison, inversion, square rooting and normalization operations in addition to the basic addition and multiplication over secret shared data.

Our proposed protocol guarantees that each client only accesses the aggregated value of the received model updates, significantly mitigating the privacy attack surface in DL even when majority of the clients are dishonest. 



Thus, the main contributions of the paper are as follows:

\begin{itemize}[topsep=0pt, itemsep=0pt]
    \item We introduce \secureDL, a novel privacy-preserving and Byzantine-resilient decentralized learning method that utilizes secure multiparty computation.
    \item We provide theoretical analysis regarding privacy guarantees offered by \secureDL~and the impact of the changes we made on the model convergence of the overall training process.
    \item We empirically compare the resilience of \secureDL~against a suite of state-of-the-art Byzantine attacks with the existing defenses over MNIST, Fashion-MNIST, SVHN and CIFAR-10 datasets. 
    \item We perform a thorough empirical analysis to assess overheads imposed by our secure aggregation protocol, to clearly characterize the cost of privacy protection on system performance in terms of computational time and network communications.
\end{itemize}

Additionally, we release the source code\footnote{The source code will be publicly available upon acceptance.} of  \secureDL~framework to stimulate further research in privacy-preserving and Byzantine-resilient decentralized learning.

\section{Related work}

Several strategies have been developed to combat Byzantine attacks in Federated training environments where a central server manages the training \cite{blanchard2017machine,chen2017distributed,guerraoui2018hidden,xie2019zeno}. These Byzantine-resistant optimization algorithms utilize robust aggregation rules to amalgamate model updates from all clients, thereby safeguarding the training against disruptions caused by malicious actors. Some methods, highlighted in \cite{pillutla2022robust} and \cite{elkordy2022heterosag}, focus on distance-based techniques. Another method involves coordinate-wise defenses, which scrutinize updates at the client coordinate level to identify anomalies. However, these can also be vulnerable to more subtle attacks \cite{baruch2019little}. Other strategies, referenced in \cite{regatti2020bygars,cao2020fltrust,xie2020zeno++}, adopt performance-based criteria. 

Distance-based methods \cite{yin2018byzantine} aim to exclude updates that significantly deviate from the benign clients' average, whereas coordinate-wise defenses look for irregularities in each coordinate of the updates. Despite their effectiveness, these strategies face challenges from sophisticated attacks, as described in \cite{baruch2019little}, where Byzantine clients craft updates that are malicious yet appear similar in distance to benign gradients. In contrast, performance-based filtering, which is employed to assess each client's model updates, often leads to more favorable results. However, it's important to note that these approaches rely on an auxiliary dataset located at the server~\cite{cao2020fltrust}.

"While significant advancements have been made in creating Byzantine-robust training algorithms for federated learning (FL), exploration of Byzantine resilience in DL has been less extensive. However, strides have been made to enhance Byzantine resilience within DL, with proposals such as ByRDiE~\cite{yang2019byrdie} and BRIDGE~\cite{fang2019bridge}, both of which adapt the Trimmed-median approach from FL to DL. ByRDiE leverages the coordinate descent optimization algorithm, while BRIDGE implements Stochastic Gradient Descent (SGD) for learning. However, \cite{guo2021byzantine} states that both approaches, ByRDiE and BRIDGE, remain susceptible to sophisticated Byzantine attacks."


In response to these vulnerabilities, Mozi \cite{guo2020towards} was developed as a Byzantine-resilient algorithm for decentralized learning. It employs a two-stage process to improve convergence, combining Euclidean distance-based selection with performance validation using training samples. Subsequently, UBAR \cite{guo2021byzantine} merges performance-based and distance-based strategies to counter Byzantine clients, with each client relying on its local dataset. This combination enables UBAR to effectively address the Byzantine attacks described in \cite{guerraoui2018hidden}. Further in this field, Basil \cite{elkordy2022basil} emerged as a fast, computationally efficient Byzantine-robust algorithm. It is distinguished by its sequential processing, memory assistance, and performance-based criteria. Designed for operation within a logical ring architecture, Basil effectively filters out Byzantine users, demonstrating robustness and operational efficiency. While performance-based algorithms like UBAR and Basil have achieved greater success compared to distance-based ones, they also introduce significant privacy risks due to their direct access to model updates. 

In our study, we have implemented the Median, Krum methods, termed \enquote*{DMedian} and \enquote*{DKrum}, Mozi and BRIDGE specifically for comparison with our proposed protocol. This choice allows us to evaluate the effectiveness of these methods against Byzantine clients within decentralized learning environments. Here, we discuss them in more details. 

\textbf{BRIDGE: } The implementation of Trimmed-Mean in BRIDGE involves excluding extreme values from the received model updates, allowing the remaining local model updates to contribute to the final average. Additionally, the network should be structured to withstand up to \( b \) Byzantine clients. Using a trim parameter \( k < \frac{n}{2} \), the server discards the highest and lowest \( k \) values, then computes the mean of the remaining \( n - 2k \) values for the global model update. To be effective against malicious clients, \( k \) must be at least equal to their number, allowing Trim-mean to manage up to nearly \( 50\% \) malicious clients.


\textbf{Mozi: } This protocol uses a distance-based metric to filter estimates based on Euclidean distance, effectively narrowing down the pool of model updates to those contributors who are most likely benign, thereby enhancing the reliability of the collective input. This is followed by a performance-based stage, where the protocol evaluates the loss of these estimates, refining the aggregation to include only those with superior performance.


\textbf{Krum: } In this approach, each client \( i \) is assigned a score \( s_i \), calculated by summing the squared Euclidean distances between its update \( \boldsymbol{w}_i \) and the updates from other clients. The scoring formula is \( s_i = \sum_{\boldsymbol{w}_j}\left\|\boldsymbol{w}_j - \boldsymbol{w}_i\right\|_2^2 \), where \( \boldsymbol{w}_j \) are the updates from a selected set of \( n-f-2 \) clients. Here, \( n \) represents the total number of clients in the network, and \( f \) is the number of Byzantine clients, which are potentially malicious clients. This selection of \( n-f-2 \) clients is crucial, as it involves choosing the updates that are closest to \( \boldsymbol{w}_i \) in Euclidean distance, effectively filtering out potential outliers or malicious updates. Krum then selects the client with the lowest score to provide the global model update, ensuring the model's resilience against adversarial actions in a scenario with up to \( f \) malicious clients.

\textbf{Median: } The Median aggregation technique, as described in \cite{yin2018byzantine}, chooses the median value of parameters as the consolidated global model, treating each parameter independently. For every \(j\)-th model parameter, the server ranks the corresponding parameters from \(m\) local models, where \( m \) is the number of participating models. The \(j\)-th parameter of the global model is then determined by the median of these ranked values:
\[ \text{Global } w_j = \text{Median} (w_{1j}, w_{2j}, \ldots, w_{mj}) .\] 


The median method can tolerate up to \( \frac{n}{2} - 1 \) Byzantine clients, where \( n \) is the total number of clients in the network. This resilience stems from the fact that as long as the majority of the values are from honest clients, the median will be derived from these honest values, ensuring the integrity of the global model despite the presence of the Byzantine clients.





\section{Preliminaries}

In this section, we provide the necessary background on decentralized learning and introduce the secure building blocks used in this paper.


\subsection{Learning in a Decentralized Setting}\label{dc}

In decentralized learning system, each client (\( c \)) from the group \( C \) directly communicates with a selected group known as their neighbors \( \mathbf{N}(c) \). These connections can be either static, formed at the beginning, or dynamic, changing over time. The network of clients forms an undirected graph \( G=\left\{C, \cup_{c \in C} \mathbf{N}(c)\right\} \), where the vertices represent clients 
and the edges indicate the connections between them\cite{lian2017can}.

In a formal setting, client \( c \) in the set \( C \) owns a private dataset $D_i=\left\{\left(x_i, y_i\right)\right\}_i$ drawn from a hidden distribution $\xi_i$ and combining the private datasets would produce the global dataset $D$ with distribution $\xi$. Each client 
begins with a shared set of model parameters denoted as \( w^0 \). 

The aim of the training process is to identify the optimal parameters, denoted as $\mathbf{w}^*$, for a machine learning model. These parameters are sought to minimize the expected loss across the entire global dataset $D$:

\begin{equation}
\mathbf{w}^* = \underset{\theta}{\arg \min } \frac{1}{|\mathbf{N}(c)|} \sum_{n_i \in \mathbf{N}(c)} \underbrace{\mathbb{E}_{s_i \sim D_i}\left[\mathcal{L}\left(\theta ; s_i\right)\right]}_{\mathcal{L}_i}\label{eqn:parameters}
\end{equation}

where for each client \( n_i \) in the network, \( \mathbb{E}_{s_i \sim D_i}\left[\mathcal{L}\left(\theta ; s_i\right)\right] \) calculates the expected loss for that client's dataset \( D_i \), where \( s_i \) is a sample from \( D_i \). Thus, the formula aims to find the model parameters \( \theta \) that minimize the average expected loss across all clients in the network.

\begin{algorithm}[ht]
\SetKwInOut{Input}{Input}
\SetKwInOut{Output}{Output}
\DontPrintSemicolon

\caption{Decentralized Learning Protocol}
\label{alg:DecentralizedLearning}

\Input{
    Initial model parameters $w_c^0$ for $c \in C$ \\
    Local training data of the client $c$: $X_c$ for $c \in C$
}
\For{$t \in [0, 1, \ldots]$}{
    \textbf{Local optimization step:}\;
    \For{$c \in C$}{
        Sample $x_c^t$ from $X_c$\;
        Update parameters: $w_c^{t+\frac{1}{2}} = w_c^t - \eta \nabla_{w_c^t}(x_c^t, w_c^t)$\;
    }
    \textbf{Communication with neighbors:}\;
    \For{$c \in C$}{
        \For{$u \in \mathbf{N}(c) \setminus \{c\}$}{
            Send $w_c^{t+\frac{1}{2}}$ to $u$\;
            Receive $w_c^{t+\frac{1}{2}}$ from $u$\;
        }
    }
    \textbf{Model updates aggregation:}\;
    \For{$c \in C$}{
        $w_c^{t+1} = \frac{1}{|\mathbf{N}(c)|} \sum_{c \in \mathbf{N}(c)} w_c^{t+\frac{1}{2}}$\;
    }
}
\end{algorithm}

This process continues through several stages until reaching a predetermined stopping point. The stages are as follows:

\begin{enumerate}[topsep=0pt, itemsep=0pt]
    \item At step $t$, each client \( c \) carries out gradient descent on their model parameters, leading to an interim model update, symbolized as \( w_c^{t+1 / 2} \).

    \item Then, the clients exchange their interim model updates \( w_c^{t+1 / 2} \) with their neighbors \( \mathbf{N}(c) \). Concurrently, they also acquire the updates from these neighbors.

    \item Finally, the clients merge the updates received from their neighbors with their own, and use this combined update to modify their local state. A basic approach to aggregation involves averaging all the updates received. Mathematically, this is represented as \( w^{t+1}=\frac{1}{\left|\mathbf{N}(c)\right|} \sum_{c \in \mathbf{N}(c)} w_c^{t+1 / 2} \).
\end{enumerate}

Algorithm \ref{alg:DecentralizedLearning} presents the complete DL training procedure.

\subsection{Secret Sharing}\label{ss}

The principle of secret sharing involves methods for distributing a secret value, denoted as \(x\), among \(n\) entities, resulting in shares \(x_1, \ldots, x_n\). Each participant \(P_i\) receives a uniformly random share \(x_i\) (under modulo \(p\)) of the concealed value over a specific \(p\). To reconstruct the original secret \(x\) from its distributed shares, the parties collaboratively pool their shares \(x_1, \ldots, x_n\) and apply a reconstruction process, involving summing the shares modulo \(p\) (\(x = (x_1 + x_2 + \cdots + x_n) \mod p\)), thereby retrieving \(x\).





Building upon this foundational concept, our work develops protocols centered on n-out-of-n additive secret sharing. We also conduct secure computations on additive shares modulo \(2^k\) within a larger ring modulo \(2^{k+s}\), where \(\sigma = s - \log(s)\) is the statistical security parameter. In cases without ambiguity, we use \([\![ x ]\!]\) to represent \([\![ x ]\!]_{2^k}\), specifically in this work where we choose \(k = 32\).



\subsubsection{Performing Addition and Multiplication on Shares}\label{mul_sec}

Various protocols have been developed to perform operations on shared secrets values. This study focuses on addition and multiplication as key functions. Our notational conventions align with those detailed in~\cite{Randmets2017ProgrammingLF}.

Secretly shared values \( [\![ x ]\!] \) and \( [\![ y ]\!] \) can be directly added as \( [\![ x ]\!] + [\![ y ]\!] = ([\![ x ]\!]_1 + [\![ y ]\!]_1, \ldots, [\![ x ]\!]_n + [\![ y ]\!]_n) \). However, multiplying secret shared values requires network interaction and distinct methodologies for each security environment.

In a dishonest majority setting, the Beaver triples technique~\cite{beaver1991efficient} is used for multiplication operations. This method involves pre-distributing shares (\([\![ a ]\!]\), \([\![ b ]\!]\), \([\![ c ]\!]\)) of a multiplication triple \((a, b, c)\), where \(a\) and \(b\) are generated uniformly randomly, and \(c = a \cdot b\). Once the shares of the inputs and the triple are received, each computation party computes \([\![ \delta ]\!]\) = \([\![ x ]\!]\) - \([\![ a ]\!]\) and \([\![ \epsilon ]\!]\) = \([\![ y ]\!]\) - \([\![ b ]\!]\) locally and then reveals \([\![ \delta ]\!]\) and \([\![ \epsilon ]\!]\) to other parties. The parties can reconstruct \(\delta\) and \(\epsilon\) using these shares. Since \(a\) and \(b\) are generated uniformly random, revealing the shares of \(\delta\) and \(\epsilon\) does not compromise the security of the protocol. Each party then locally computes:
\begin{equation*}
    [\![ w ]\!]_i = [\![ c ]\!]_i + \epsilon \cdot [\![ b ]\!]_i + \delta \cdot [\![ a ]\!]_i + \epsilon \cdot \delta
\end{equation*}
where \( [\![ w ]\!]_i \) is a share of the multiplication result calculated by computation party \textit{i}.


\subsubsection{Secure Comparison Protocol}\label{seccom}

In the context of secret sharing, a secure comparison protocol enables the comparison of additive shares without revealing the underlying private values. It determines if one value is greater, lesser, or equal to another, producing a binary result that preserves the confidentiality of the inputs. This protocol is crucial for privacy-sensitive applications, ensuring that only the outcome of the comparison is disclosed, thereby maintaining the integrity and secrecy of each party's input.

In this work, we have implemented the secure comparison protocol as described by Makri et al. \cite{makri2021rabbit}. Their protocol ensures perfect security over the arithmetic ring $\mathbb{Z}_M$ for positive integers in the range $\left[0,2^l-1\right]$, where $l \in \mathbb{N}$ and $2^{l+1} < M = 2^k$. The protocol's adaptability to different arithmetic models and its effectiveness in settings with a dishonest majority are its key strengths.

The basic idea of the protocol is that the sum of two secret shared values $a, b \in \mathbb{Z}_M$ modulo $M$ is less than $a$ and less than (LT) $b$, iff $a+b$ is reduced by $M$ :

\begin{align*}
(a+b) \bmod M &= a+b-M \cdot \mathrm{LT} (a+b \bmod M, a) \\
&= a+b-M \cdot \mathrm{LT}(a+b \bmod M, b)
\end{align*}




The protocol's online phase involves minimal rounds and arithmetic operations, while the offline phase requires additional resources such as edaBits \cite{escudero2020improved}.

\subsubsection{Inversion Operation} \label{inv}

Our multi-party computation framework does not support inversion operations natively; therefore, we adopt an iterative approximation approach, akin to the method proposed by Nardi~\cite{nardi2012achieving}. This strategy enables us to facilitate inversion using basic operations such as addition and multiplication, effectively circumventing the need for direct inversion over secret share inputs.

The core concept revolves around identifying a value \( B \) that serves as the inverse of a specific value \( X \). To achieve this, we define a function \( f(x) \) aimed at determining \( B \) by setting \( f(x) = 0 \) when \( x = B \), essentially making \( B \) the solution that satisfies the inverse relationship for \( X \). The function is formulated as follows:

\begin{equation}
f(x) = X^{-1} - B \nonumber
\end{equation}

For determining the root of \( f(x) \), the Newton-Raphson method~\cite{agresti2003categorical} is suggested. This results in a stable numerical iterative approximation represented as:

\begin{align*}
    B_{s+1} = 2B_s - B_sM_s \quad B_0 = c^{-1} \\
    M_{s+1} = 2M_s - M_s^2 \quad  M_0 = c^{-1}X
\end{align*}

where \(c\) is a constant, and \(M_s = B_s \times A\), with \(B_s\) converging to the approximate inverse of \(X\) after about 15 iterations~\cite{ghavamipour2022privacy}. 





\subsubsection{Square Root Computation}\label{sqr}

In our protocol, computing square roots is essential but not directly feasible in our implemented multi-party computation framework due to the lack of native support. We overcome this by employing an approximation technique rooted in the Newton-Raphson method, enabling square root calculation through just multiplication and addition~\cite{press2007numerical}. 

The Newton-Raphson method refines an initial approximation of a square root through an iterative process involving a function and its derivative. For computing the square root of a number \( Y \), the method utilizes the function \( f(x) = x^2 - Y \), aiming to identify an \( x \) that satisfies \( f(x) = 0 \). The method updates the approximation iteratively as follows:

\begin{align*}
    B_{s+1} = 2B_s - B_sM_s \quad B_0 = c^{-1} \\
    M_{s+1} = 2M_s - M_s^2 \quad  M_0 = c^{-1}X
\end{align*}

where \( x_n \) is the current guess and \( x_{n+1} \) is the updated guess. Starting from an initial guess \( x_0 \), the method iteratively applies this update. The quotient \( \frac{Y}{x_n} \) offers an increasingly accurate approximation of the square root of \( Y \) when averaged with \( x_n \). After a sufficient number of iterations (15 iterations in our implementation), \( x_n \) converges to the square root of \( Y \). 





\subsubsection{Preprocessing}\label{sec:preprocess}

Our protocol is designed to be compatible with Message Authentication Codes (MACs) and Beaver triples generation, following approaches similar to those in \cite{keller2020mp}. Our implementation necessitates multiplication triples over \(\mathbb{Z}_{2^k}\), where \(k\) is a parameter of the protocol. For the secure comparison protocol, when \(k=1\), our protocol incorporates an optimized variant of the two-party TinyOT protocol \cite{wang2017global}. For larger choices of \(k\) ( \(k=32\) in this work), we base our costs on the MASCOT protocol \cite{keller2016mascot}, noting that despite MASCOT's communication complexity being in \(O(k^2)\), it still offers the lowest costs for all the table sizes we have considered, particularly with \(k=32\). 

In our analysis, we concentrate primarily on the dynamics of the online phase of the protocol. We omit the details of the offline phase, encompassing the generation and distribution of MACs and Beaver triples, which can be found in \cite{CramerDESX18}.

\section{Threat model}


We assume that two distinct entities are involved in the training: semi-honest and malicious parties:

\textbf{Semi-honest clients: } Semi-honest clients are fundamentally non-malicious, but exhibit an elevated level of curiosity (also known as Honest-but-Curious clients). These clients strictly follow the established protocols of the system but try to infer additional information about other clients from the exchanged messages. 
This behavior, while not explicitly harmful, raises genuine concerns regarding data privacy and confidentiality within the system \cite{pasquini2023security}.

\textbf{Byzantine clients: } The second type of entities are Byzantine clients who are defined by their inherent malevolence. Unlike semi-honest clients, Byzantines display no commitment to the established protocols and are known for their unpredictable malicious actions. Their objective is to disrupt the learning process or introduce biases into the model that is being trained. They may do this by transmitting arbitrary, misleading, or false information to other clients. Furthermore, multiple Byzantine clients may collude to maximize the efficacy of either poisoning or inference attacks. 

\subsection{Problem Statement}

In all collaborative learning environments, an honest client is not interested in acquiring a specific model from another client. Rather, its objective is to obtain an aggregated model update derived from a combination of multiple local client datasets. Having access only to the aggregated model update offers two crucial privacy benefits. Firstly, it prevents malicious attempts targeting client models by masking details specific to individual local models, effectively hindering adversaries from identifying and exploiting precise data sources. Secondly, it diminishes the impact of any single client's input by combining contributions from multiple clients within the aggregated model. This reduction in the prominence of individual contributions serves to obscure their visibility, thus decreasing the likelihood of attacks that target unique datasets.

In FL, a central server aggregates the model updates received from the clients, relaying only the aggregated outcome to those clients. Thus, the server is the only entity with access to each model's (non-aggregated) updates. Conversely, in DL, with a shift to a completely decentralized training, the elimination of the central server increases the attack surface. This increase is due to the need to distribute updates across neighboring clients for on-client aggregation, contrasting sharply with FL. Therefore, in DL, each client, with access to collaborator updates can potentially perform privacy attacks against their collaborators such as inference attacks~\cite{pasquini2023security}.



Moreover, the susceptibility of collaborative learning systems to Byzantine client disruptions is a significant concern. In these systems, a single Byzantine client has the potential to derail the entire learning process by manipulating or corrupting shared models across the network \cite{blanchard2017machine, raynal2023can}. In FL settings, various research studies have been conducted to enable the server to detect and mitigate Byzantine threats. However, in DL, the absence of centralized control within decentralized architectures presents unique challenges. In a decentralized setting, each client in the network must independently address these challenges, which complicates the detection and mitigation of malicious activities. 

In this study, our objective is to develop a DL method that maintains Byzantine robustness against malicious clients without compromising accuracy or efficiency. The Byzantine robustness of our protocol is achieved through a privacy-preserving aggregation rule. This rule prevents clients from accessing other clients' data in plain form; instead, only aggregated values are revealed to them.

\section{\secureDL}

This section presents a detailed explanation of the secure aggregation rule within \secureDL. The approach we take for aggregation is illustrated in Figure \ref{Vector}. 


Given the existence of Byzantine clients within the network, it becomes essential to acknowledge the capability of compromised clients to prevent the convergence of the global model. In light of this, our proposed aggregation rule underlines the significance of evaluating both the magnitude and the directional alignment of local model updates concerning the global model of each client.

Byzantine clients manipulate model update directions, steering the global model away from its intended trajectory. This strategic deviation impedes the model's convergence. The resulting angular disparity between the model updates from these malicious clients and the global model of the receivers of those updates exceeds that of benign clients. In \secureDL, the direction of received model updates is evaluated against the global model of the receiving client using the \textbf{cosine similarity} metric.
This metric is widely recognized for its ability to compute the angular difference between two vectors. The cosine similarity between these two vectors \( \mathbf{w}_a \) and \( \mathbf{w}_b \) is expressed as:
\begin{equation*}
    Cosine(\mathbf{w}_a, \mathbf{w}_b) = \frac{\mathbf{w}_a \cdot \mathbf{w}_b}{\|\mathbf{w}_a\| \|\mathbf{w}_b\|}.
\end{equation*}

In our protocol, we tackle the challenge of computing cosine similarity in a decentralized learning environment while preserving the privacy of the input vectors. Here, direct access to complete model updates is restricted, and clients should only have access to the aggregated model updates, necessitating an innovative approach for accurate similarity measurement. Our solution leverages collaborative client efforts to execute secure vectorized operations, alongside specialized computations for square roots and inversions, to calculate cosine similarities between additive secret-shared model update vectors $[\![\mathbf{w}_a]\!]$ and $[\![\mathbf{w}_b]\!]$.

The computation of cosine similarity within our secure framework involves several complex operations, as detailed in Algorithm \ref{alg:CosineSimilarityIntegratedNorm}. Initially, we utilize the Beaver multiplication protocol for securely computing the dot product between two model updates. This crucial step allows for the multiplication of values without directly revealing them to any party. For the calculation of each model update's norm, we square each vector's elements, sum these squares, and compute the square root of this sum, a process elaborated in Section \ref{sqr} for square root computation. Following this, leveraging the inversion function introduced in Section \ref{inv}, we compute the inversion of the product of these norms. This sequence of operations—dot product computation, norm calculation, and inversion—ensures the accurate and secure calculation of cosine similarity, adhering to the privacy and security requirements outlined in our protocol.

\begin{figure}[t]
\centering
  \includegraphics[width=\linewidth]{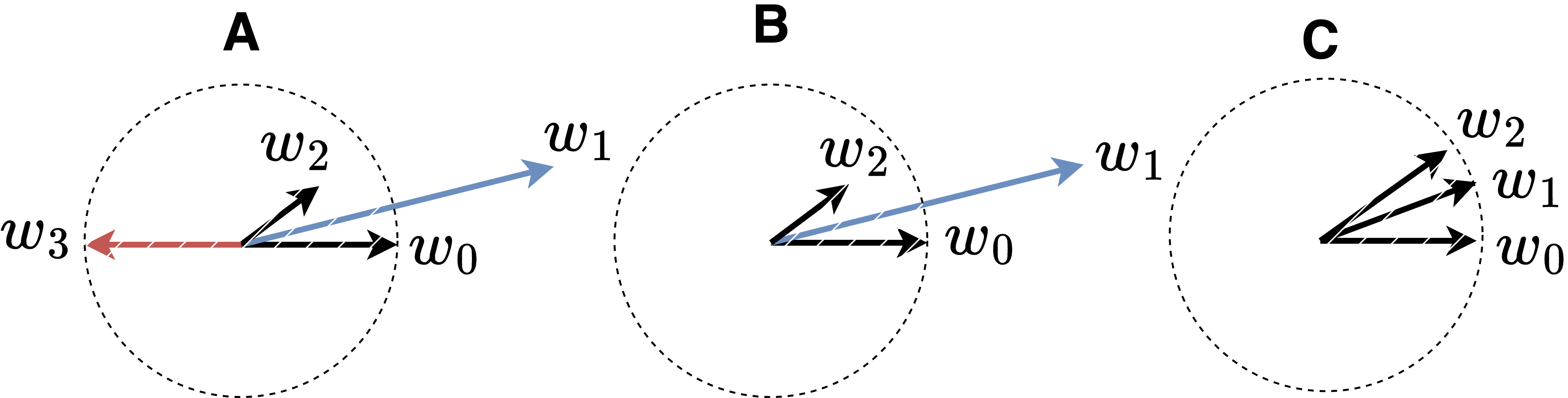}
  \caption{The illustration of \secureDL~aggregation rule. The model update of $Client_{0}$, denoted as $w_0$, obtains the updates from other clients (\(w_1\), \(w_2\), and \(w_3\)). In image \textbf{A}, the protocol discards any model update with a negative cosine value. Image \textbf{B} depicts $Client_0$ normalizing the received models based on its own model's magnitude. Image\textbf{C} shows the final accepted model update after normalization.}
  \label{Vector}
\end{figure}

\begin{algorithm}
\SetKwInOut{Input}{Input}
\SetKwInOut{Output}{Output}
\DontPrintSemicolon

\caption{Cosine Similarity Computation}
\label{alg:CosineSimilarityIntegratedNorm}

\Input{
    Two secret-shared vectors $[\![\mathbf{w}_a]\!]$ and $[\![\mathbf{w}_b]\!]$.
}
\Output{
    The cosine similarity $[\![\text{cosine}]\!]$ between the vectors $[\![\mathbf{w}_a]\!]$ and $[\![\mathbf{w}_b]\!]$.
}
\BlankLine

$[\![\text{dotProduct}]\!] \leftarrow [\![\mathbf{w}_a]\!] \odot [\![\mathbf{w}_b]\!]$ \tcp*{Vector dot product}

\tcp{Compute the norm of vector $\mathbf{w}_a$ using vectorized operations}
$[\![\mathbf{w}_a^2]\!] \leftarrow [\![\mathbf{w}_a]\!] \times [\![\mathbf{w}_a]\!]$ \\
$[\![\text{sum}_a]\!] \leftarrow \sum_{i=1}^{n} [\![w_{ai}^2]\!]$ \\
$[\![\text{norm}_a]\!] \leftarrow \text{ComputeSquareRoot}([\![\text{sum}_a]\!])$ \\

\tcp{Compute the norm of vector $\mathbf{w}_b$ using vectorized operations}
$[\![\mathbf{w}_b^2]\!] \leftarrow [\![\mathbf{w}_b]\!] \times [\![\mathbf{w}_b]\!]$ \\
$[\![\text{sum}_b]\!] \leftarrow \sum_{i=1}^{n} [\![w_{bi}^2]\!]$ \\
$[\![\text{norm}_b]\!] \leftarrow \text{ComputeSquareRoot}([\![\text{sum}_b]\!])$ \\

$[\![\text{denom}]\!] \leftarrow \text{ComputeInverse}([\![\text{norm}_a]\!] \times [\![\text{norm}_b]\!])$ \tcp*{Element-wise multiplication for norms and then inversion}
$[\![\text{cosine}]\!] \leftarrow [\![\text{dotProduct}]\!] \times [\![\text{denom}]\!]$ \\

\Return{$[\![\text{cosine}]\!]$}
\end{algorithm}

This streamlined computation facilitates the efficient isolation of model updates that fall below a certain similarity threshold, thereby enhancing the robustness of the decentralized learning process. It is important to highlight that the results of the cosine similarity computation remain in a secret-shared form. This necessitates a method for comparing these results with a threshold without reconstructing the secret-shared value. To address this, we employ a secure comparison protocol, as described in Section~\ref{seccom}, which securely enables the decision to accept or discard a model update. This protocol allows for the comparison of secret-shared cosine similarity scores against a predefined threshold in a manner that upholds the privacy and integrity of the data. Our method not only meets the secure framework's privacy requirements but also facilitates a more effective evaluation of model update contributions, ensuring the inclusion of only those that are closely aligned with the overall model direction, without having direct access to the model updates.


Moreover, Byzantine adversaries have the ability to significantly magnify the magnitudes of local model updates in compromised clients, thereby exerting a substantial impact on the receiver's global model update. To mitigate this, our approach involves the \textbf{normalization} of each received model update, aligning it with the magnitude of the receiver's model update. This process involves adjusting the scale of received model updates to ensure their magnitudes match with that of the receiver client's model update within the vector space. To achieve this normalization effectively, we employ L2-normalization techniques.

The L2-normalization of model update \( \mathbf{w}_a \) with respect to receiver \( \mathbf{w}_b \) can be mathematically expressed as:
\[ Norm(\mathbf{w}_a, \mathbf{w}_b) = \mathbf{w}_a \times \frac{\lVert \mathbf{w}_b \rVert_2}{\lVert \mathbf{w}_a \rVert_2} \]
where \( \lVert \mathbf{w}_b \rVert_2 \) and \( \lVert \mathbf{w}_a \rVert_2 \) represent the L2 norms (or Euclidean norms) of the model update vectors \( \mathbf{w}_b \) and \( \mathbf{w}_a \), respectively. 
This normalization method balances the effect of each local model update on the overall global model. It lessens the influence of larger updates and increases the impact of smaller ones. By giving more weight to these smaller updates, we can reduce the effect of harmful updates, leading to a stronger global model.

In our secure computation framework, we achieve privacy-preserving L2 normalization using accurate approximations. In the first step, as outlined in the Algorithm \ref{ComputeL2Normalization}, clients jointly and securely calculate the norms of the reference model update and received model update vectors. This process begins with squaring each vector component, summing these squares, and then securely computing the square root to obtain the vectors' magnitudes in Euclidean space. The procedure emphasizes the collaborative effort and complexity inherent in our framework, highlighting the advanced protocols employed to ensure privacy.

To normalize one vector relative to another, we use the inversion protocol to compute the inverse of one vector's norm. This step is crucial for adjusting the magnitude of the received model update to match that of the receiver clients, maintaining consistency with the client's global model without revealing private information. The final normalization step involves scaling the original vector by a computed ratio, ensuring the resulting vector maintains its direction. At the same time, its magnitude is adjusted relative to the reference vector.

\begin{algorithm}
\SetKwInOut{Input}{Input}
\SetKwInOut{Output}{Output}
\DontPrintSemicolon
\caption{ComputeL2Normalization}
\label{ComputeL2Normalization}

\Input{
    Secret-shared vectors \( [\![\mathbf{w}_a]\!] \) and \( [\![\mathbf{w}_b]\!] \).
}
\Output{
    L2-normalized vector \( [\![\mathbf{w}_a']\!] \) relative to \( [\![\mathbf{w}_b]\!] \).
}

\BlankLine
\tcp{Vector norm computation}
\( [\![\mathbf{w}_a^2]\!] \leftarrow [\![\mathbf{w}_a]\!] \times [\![\mathbf{w}_a]\!] \) \;
\( [\![\text{sum}_a]\!] \leftarrow \sum_{i=1}^{n} [\![w_{ai}^2]\!] \) \;
\( [\![\lVert \mathbf{w}_a \rVert_2]\!] \leftarrow \text{ComputeSquareRoot}([\![\text{sum}_a]\!]) \) \;

\tcp{Vector norm computation}
\( [\![\mathbf{w}_b^2]\!] \leftarrow [\![\mathbf{w}_b]\!] \times [\![\mathbf{w}_b]\!] \) \;
\( [\![\text{sum}_b]\!] \leftarrow \sum_{i=1}^{n} [\![w_{bi}^2]\!] \) \;
\( [\![\lVert \mathbf{w}_b \rVert_2]\!] \leftarrow \text{ComputeSquareRoot}([\![\text{sum}_b]\!]) \) \;

\( [\![\lVert \mathbf{w}_a \rVert_2^{-1}]\!] \) \(\leftarrow\) ComputeInverse(\( [\![\lVert \mathbf{w}_a \rVert_2]\!] \))\;

\tcp{Compute the ratio for normalization}
\( [\![r]\!] \) \(\leftarrow\) \( [\![\lVert \mathbf{w}_b \rVert_2]\!] \) \(\times\) \( [\![\lVert \mathbf{w}_a \rVert_2^{-1}]\!] \)\;

\tcp{Normalize the vector \( [\![\mathbf{w}_a]\!] \) by the computed ratio}
\( [\![\mathbf{w}_a']\!] \) \(\leftarrow\) \( [\![\mathbf{w}_a]\!] \) \(\times\) \( [\![r]\!] \)\;

\Return{\( [\![\mathbf{w}_a']\!] \)}\;
\end{algorithm}

The complete \secureDL~protocol is illustrated in Algorithm \ref{alg:SecureDL} where the concluding step involves computing the mean of the normalized and global model of the receiver clients. This entire process is repetitively applied across $R_g$ global iterations to finalize the training process. A local learning rate $\beta$, and a global learning rate $\alpha$ defines the relative size of the updates for local update computation, and global aggregation respectively. Rejection of updates is controlled by a configurable threshold $tau$.

\begin{algorithm}[htb!]
\SetKwInOut{Input}{Input}
\SetKwInOut{Output}{Output}
\DontPrintSemicolon

\caption{\secureDL}
\label{alg:SecureDL}

\Input{
    $n$ clients $C_0, C_1, \dots, C_{n-1}$ with local training datasets $D_i$ for $i = 0, 1, \dots, n-1$ \\
    Number of global iterations: $R_g$ \\
    Threshold: $0 < \tau < 1$ \\
}
\Output{Global model $\mathbf{w}$ trained on all clients' datasets for each client}

\textbf{Initialization step:}\;
All clients start from the same initial model $w^0$.\;
\For{$r = 1,2, \dots, R_g$}{

    \textbf{Local training step:}\;
    Training local models\;
    \For{every client $i$ in $\{C_0, C_1, \dots, C_{n-1}\}$}{
        $w_i = \text{LocalUpdate} (w^r, D_i, \beta)$\;
        Send model update shares $\llbracket \mathbf{w}_i \rrbracket$ to all clients\;
    }
    \textbf{Aggregation step:}\;
    Updating their global model via aggregating the local model updates\;
    \For{every client $i$ in $\{C_0, C_1, \dots, C_{n-1}\}$}{
        \For{every client $j$ in $\{C_0, C_1, \dots, C_{n-1}\}$}{
            \If{$i \neq j$}{
                $\llbracket \boldsymbol{cosine}_{ij} \rrbracket \leftarrow \text{Cosine}(\llbracket \mathbf{w}_{i} \rrbracket, \llbracket \mathbf{w}_{j} \rrbracket) = \frac{\langle \llbracket \mathbf{w}_i \rrbracket, \llbracket \mathbf{w}_{j} \rrbracket \rangle}{\|\llbracket \mathbf{w}_i \rrbracket\| \cdot \|\llbracket \mathbf{w}_{j} \rrbracket\|}$
            }

        \If{$\text{SecureComparison}(\llbracket \boldsymbol{cosine_{ij}} \rrbracket , \tau)$}{
            \tcp{Returns True if $cosine_{ij} < \tau$}
            $\llbracket \overline{w_j} \rrbracket = 0$\;
        }
        \Else{
            $\llbracket \overline{\mathbf{w}_j} \rrbracket =  Normalization(\llbracket \mathbf{w}_{i} \rrbracket, \llbracket \mathbf{w}_{j} \rrbracket) = \frac{\|\llbracket \mathbf{w}_{j} \rrbracket\|}{\|\llbracket \mathbf{w}_i \rrbracket\|} \cdot \llbracket \mathbf{w}_i \rrbracket$\;
        }}
    \textbf{Model updates aggregation:}\;
{
    $\llbracket \hat{\mathbf{w}} \rrbracket = \frac{1}{n} \left( \sum\limits_{\substack{j=0 \\ j \neq i}}^{n-1} \llbracket \overline{\mathbf{w}_j} \rrbracket + \llbracket \mathbf{w}_{i} \rrbracket \right)$\;
}

    }
}
\end{algorithm}

\subsection{Convergence} 
\label{sub:convergence}
In this section, we demonstrate a proof that \secureDL~converges. Specifically, for an arbitrary number of malicious updates, the difference between the model learnt from \secureDL~filtered updates and the optimal model under no attacks is bounded. 
Our proof closely follows that presented in FLtrust \cite{cao2020fltrust}, as beyond the secure multi-party computation, the rejection of malicious updates on shares of updates results in the overall computation of cosine similarity with a known `good update' and each received client's update. Collectively, the clients find a model $w$ from a global dataset $D$ defined previously in Equation~\ref{eqn:parameters}, where $D = \cup^{n}_{i=1}D_i$ is composed of the joint training datasets from each decentralized client on which the expected loss function $f(D, \mathbf{w})$ is minimized. 
We define $\mathbf{w}^*$ as the optimal model for a given data distribution $\xi$.


We present several assumptions stated in similar cosine similarity based detection schemes~\cite{cao2020fltrust, miao2022privacy} with which \secureDL~adheres to satisfy Theorem 1:
\begin{assumption}
    The expected loss function $f(D, \mathbf{w})$ is $\mu$-strongly convex and differentiable over $\mathbb{R}^d$ with L-Lipschitz continuous gradient.

    This entails that for $\mathbf{w}, \mathbf{\widehat{w}} \in \mathbb{R}^d$:
    $$
    \nabla f(D, \mathbf{\hat{w}}) \ge \nabla f(D, \mathbf{w}) + \langle \nabla f(D, \mathbf{w}), \mathbf{\hat{w}} - \mathbf{w} \rangle) + \frac{\mu}{2}\| \mathbf{\hat{w}} - \mathbf{w} \|^2,
    $$
    $$
    \|\nabla f(D, \mathbf{w}) - \nabla f(D, \mathbf{\widehat{w}})\| \leq L'\|\mathbf{w} - \mathbf{\widehat{w}}\|,
    $$
    where $\nabla$ represents the gradient and $\|\cdot\|$ represents the L2 norm. and $\langle\cdot,\cdot\rangle$ represents the inner product of two vectors. 
    Further, the empirical loss function $f(D, \mathbf{w})$ is L1-Lipschitz probabilistic.
\end{assumption}

\begin{assumption}
    The gradient of the empirical loss function $\nabla f(D, \mathbf{w})$, with respect to $\mathbf{w}$, is bounded. Additionally, the gradient difference for any $\mathbf{w} \in \mathbb{R}^d$ is bounded.
\end{assumption}

\begin{assumption}
    Each client's local training dataset $D_i:=\{\xi_k\}_{k=1}^n$ is independently drawn from the distribution $\xi$.
\end{assumption}


\begin{assumption}
    Approximations in the secure comparison and secure aggregation that result in minute variations in update values are significantly smaller than expected noise in local dataset variation, thereby negligibly impacting learning.
\end{assumption}

\begin{theorem}
\label{convergenceproof}
    Suppose Assumption 1-4 hold and \secureDL~uses $R_g=1$ and $\beta=1$. For an arbitrary number of malicious clients, the difference between the global model learnt by \secureDL~and the optimal global model $\mathbf{w^*}$ under no attacks is bounded. Formally, we have the following with probability at least $1-\delta$:
    $$\|\mathbf{w^t} - \mathbf{w^*}\| \leq (1-\rho)^{t} \|\mathbf{w^0} - \mathbf{w^*}\| + 12 \alpha \Delta_1 / \rho,$$
    where $\mathbf{w^t}$ is the model at the $t^{th}$ iteration, $\alpha$ is the global learning rate, and both $\rho$, $\Delta_1$ as constants defined in \cite{cao2020fltrust}.
    
    Notably when $|1-\rho| < 1$, with a large number of iterations, $\lim_{t\rightarrow\infty} \|\mathbf{w^t} - \mathbf{w^*}\| \leq 12\alpha\Delta_1 / \rho$, and thus bounded.
    \label{thm:convergence}
\end{theorem}

\begin{proof}[Proof of Theorem \ref{convergenceproof}]
With the assumptions mentioned in Section \ref{sub:convergence} and Theorem~\ref{thm:convergence}, adopted for \textbf{SecureDL} from \cite{cao2020fltrust}, we can demonstrate \secureDL's model under attack and the optimal model is bounded. We refer the reader to \cite{cao2020fltrust} for the full proof of Theorem~\ref{thm:convergence} through Lemma 1, 2 and 3. We note that when considering this proof for \secureDL's update rejection, the aggregation criteria mirrors that of FLtrust albeit with variation to the gradient scaling parameter $\varphi$ in Lemma 1. Specifically, FLtrust scales each normalized gradient $\bar{g_i}$ by the cosine similarity $cosine_i$ for a given client update $i$ in set $\mathcal{S}$, whose $cosine_i$ is positive, As such FLtrust's scaling parameter is $\varphi_i = \frac{ReLu(cosine_i)}{\sum_{j\in\mathcal{S}} ReLu(cosine_j)}$, where $ReLu$ is the rectified linear unit activation function. 
On the other hand, in \secureDL, the gradient is not scaled by $cosine_i$, instead retaining it's normalized magnitude of 1, better represented by the binary $step$ activation above a given threshold $\tau$, as such, $\varphi_i = \frac{Step(cosine_i, \tau)}{|\mathcal{S}|}$, where $0< \tau < 1$, thereby satisfying $0 < \varphi_i < 1$, and $\displaystyle \sum_{j\in\mathcal{S}}\varphi=1$. With this variation on Lemma 1, \secureDL~satisfies Lemma 
1, 2, 3, and Theorem~\ref{thm:convergence} remains valid for \secureDL, and thus the model obtained by \secureDL~shall converge. While we refer the reader to \cite{cao2020fltrust} for the full proof of Theorem, we shall in the next section empirically demonstrate SecureDL's convergence over 3 distinct datasets.
\end{proof}






\subsection{Privacy analysis}

In this section, we show that \secureDL~protocol achieves privacy-preserving aggregation. We provide the privacy analysis using the simulation based proof technique~\cite{goldreich1998secure,lindell2017simulate}. In simulation based proof, the privacy is guaranteed by showing that there is a polynomial-time simulator that simulates the view of the adversary such that the adversary cannot distinguish the simulated view from the real one. 
The (computational) indistinguishability ensures that the adversary does not learn any information about the honest clients since even the simulator does not have any private information belongs to the honest parties.
Below, we present an informal statement of this principle; the formal analysis is given in Appendix \ref{sec:privacyAnalysis}.

\begin{theorem}[Privacy w.r.t Semi-Honest Behavior]
\label{securedlproof}
Consider an \(n\)-party secure multi-party computation protocol \(\Pi\), designed to compute a function \(f\) over inputs \(\bar{x}:=\left(x_1, \ldots, x_n\right)\) utilizing the \secureDL~algorithm. The protocol involves \(n\) parties, \(P_1, P_2, \ldots, P_n\), each holding an additive share of their private inputs. The protocol \(\Pi\) is said to privately compute the function \(f(\bar{x})\) in the presence of at least one honest party and potentially colluding parties, if for every subset of honest parties \(I \subseteq [n]\), with \(I=\left\{i_1, \ldots, i_t\right\}\), there exists a probabilistic polynomial-time simulator \(S_I\). The output distribution of \(S_I\) and the actual view of the parties in \(I\) during the protocol's execution, \(\operatorname{View}_I^{\Pi}(\bar{x})\), are required to be computationally indistinguishable (\(\stackrel{\text{perf}}{\equiv}\)) across all possible inputs \(\bar{x}\), where:

\[
\{S_I(\bar{x}, f_I(\bar{x}))\}_{\bar{x}} \stackrel{\text{perf}}{\equiv} \{\operatorname{View}_I^{\Pi}(\bar{x})\}_{\bar{x}}.
\]

where \(f_I(\bar{x})\) denotes the subsequence of function outputs relevant to the subset \(I\), specifically \(f_{i_1}(\bar{x}), \cdots, f_{i_t}(\bar{x})\), and \(\operatorname{View}_I^{\Pi}(\bar{x})\) represents the collective view of the subset \(I\) during the execution of \(\Pi\), encompassing all information accessible to the parties in \(I\).
\end{theorem}

\secureDL~protocol achieves the privacy properties (under the semi-honest model) by leveraging information theoretical security of the secret sharing scheme in the aggregation method. It allows the construction of a simulation that accurately mimics the adversary's view without necessitating access to the honest parties' inputs. This ensures that the protocol can simulate the adversary's perspective in a manner that is computationally indistinguishable from the actual computation, thereby maintaining the privacy of the inputs against passive adversaries.

In the \secureDL~protocol, we utilize several operations include multiplication, addition, inversion, square root calculation, vector norm computation, cosine similarity assessment, and L2 normalization. For each of these functions, we have provided simulation-based proofs individually. The utilization of randomness inherent in the secret sharing mechanism enables the secure and consistent simulation of each component's composition within the protocol. This critical feature ensures that the output shares generated by each function—despite being reused as inputs for subsequent operations—remain indistinguishable to potential adversaries. Thus, the privacy and confidentiality of data through the computation process are maintained, effectively protecting against unauthorized disclosure and preserving the privacy of the computational outcomes. Finally, we are able to provide a simulator that simulates the view of an adversary (who can control up to $n-1$ clients) for the whole protocol, and prove Theorem~\ref{securedlproof}.

\textbf{Strengthening from semi-honest to Byzantine clients:} Our implemented MPC protocol is designed to provide security against semi-honest adversaries within a dishonest majority setting over the ring $\mathbb{Z}_{2^k}$, creating a robust framework for ensuring the privacy of computations against passive attackers. 
As mentioned in Section~\ref{sec:preprocess}, by using information-theoretic secure and homomorphic Message Authentication Codes (MACs) given in~\cite{CramerDESX18}, we can prevent Byzantine adversaries to fabricate inconsistent information about their (preprocessing) data.

\section{Evaluation}

We conducted a set of  experiments to demonstrate the robustness of \secureDL~in the case of Byzantine adversaries and to measure the overhead introduced by the changes to the training process. In what follows, we present the experimental details and discuss the findings.

\subsection{Experimental Setup}

We implement SecureDL in Python leveraging the Pytorch framework~\cite{paszke2019pytorch} for model training. All experiments are conducted on a compute cluster with a Intel Xeon Platinum 8358 cpu, A100 GPU, and 128GB RAM networked locally. 

\subsubsection{Datasets} We evaluated our protocol using four popular deep learning datasets: MNIST~\cite{lecun1998mnist}, Fashion-MNIST~\cite{xiao2017fashion}, SVHN~\cite{netzer2011reading} and CIFAR-10~\cite{krizhevsky2009learning} with independent and identically distributed (IID) partitioning. We randomly split the training set into subsets of a desired size (mostly evenly) and allocate each subset to each client. The MNIST dataset contains 60,000 training and 10,000 testing grayscale images of handwritten digits, each with a resolution of \(28 \times 28\) pixels. It is evenly distributed across ten classes, with each class represented by 6,000 images in the training set and 1,000 in the test set. Fashion-MNIST is a 10-class fashion image classification task, which has a predefined training set of 60,000 fashion images and a testing set of 10,000 fashion images. The SVHN dataset, sourced from Google Street View house numbers, includes 99,289 color images across ten classes. It comprises of 73,257 training and 26,032 testing images, all standardized to a resolution of \(28 \times 28\) pixels. CIFAR-10, a collection of color images, offers a diverse challenge with its 50,000 training and 10,000 testing examples spread across ten distinct classes. Each example in this dataset is a color image representing one of these classes.

\subsubsection{Model}

We implemented two distinct model architectures to cater to the specific needs of our datasets: a CNN for CIFAR-10 and SVHN, and an MLP for Fashion MNIST and MNIST. Detailed in Table \ref{model_arc}, the CNN architecture begins with a \(3 \times 32 \times 32\) input layer, progressing through several convolutional layers with Group Normalization (GN) and ReLU activations, interspersed with max pooling and dropout rates ranging from 0.2 to 0.5, and concludes with a fully connected output layer. This model is optimized with a learning rate of 0.002 and a batch size of 128, designed to efficiently process and learn from the image data. For the simpler Fashion MNIST and MNIST datasets, the two-layer MLP, employing sigmoid activations, is optimized with a learning rate of 0.01 and the same batch size, providing a tailored solution that addresses the unique challenges of each dataset while ensuring effective learning outcomes.

\begin{table}
\caption{Model architecture for CIFAR10 dataset}
\label{model_arc}
\centering
\resizebox{1.0\columnwidth}{!}{%
\begin{tabular}{|c|c|}
\hline Layer Type & Size / Type \\
\hline Input & $3 \times 32 \times 32$ \\
\hline Conv+ GroupNorm + ReLU & $3 \times 3 \times 32$ / GN(32) \\
\hline Conv + GroupNorm + ReLU + Max Pool + Dropout(0.2) & $3 \times 3 \times 32$ / GN(32) \\
\hline Conv + GroupNorm + ReLU & $3 \times 3 \times 64$ / GN(64) \\
\hline Conv + GroupNorm + ReLU + Max Pool + Dropout(0.3) & $3 \times 3 \times 64$ / GN(64) \\
\hline Conv + GroupNorm + ReLU & $3 \times 3 \times 128$ / GN(128) \\
\hline Conv + GroupNorm + ReLU + Max Pool + Dropout(0.4) & $3 \times 3 \times 128$ / GN(128) \\
\hline Fully Connected + GroupNorm + ReLU + Dropout(0.5) & 128 / GN(128) \\
\hline Fully Connected & 10 \\
\hline
\end{tabular}
}
\end{table}

\subsection{Attacks}

Our experimental evaluation includes two primary categories of poisoning strategies: data poisoning and local model manipulation. Within the realm of data poisoning, our focus is on the widely recognized label flipping attack \cite{jere2020taxonomy}. In the context of local model manipulation, we explore various approaches, including the sign flipping attack \cite{karimireddy2021learning}, Gaussian attack \cite{li2021byzantine}, and the scaling attack \cite{bagdasaryan2020backdoor}. We also consider scenarios that involve the mix of these (different) attacks to understand their cumulative impact.

\textbf{Sign-flipping attack (SF): }
In order to implement this attack, we apply a reversal by multiplying each update from a malicious client with \(-1\), effectively flipping the update vector's direction. Thus, if the original update vector is \( \hat{w}_i^k \), the transmitted update becomes \( w_i^k = -1 \times \hat{w}_i^k \). This attack disrupts the gradient descent process and emphasizes the need for robust detection/protection strategies against such manipulations.

\textbf{Gaussian attack (noise): }
Referred also as the noise attack, a Byzantine client subtly interferes with the learning process in this attack by changing its updates, \( w_i^k \), to follow a Gaussian distribution. For our experiments, the Gaussian distribution is characterized by a mean and variance of $0.1$, simulating noise to evaluate the learning algorithm's resilience against these statistically designed disruptions.

\textbf{Scaling attack (SA): }
In this adversarial approach, a Byzantine client dramatically increases the magnitude of its model's weights prior to aggregation. This action distorts the global model by injecting excessively large weights, which, when averaged with the updates from benign clients, can result in model divergence or significant skew. 

\textbf{Label Flipping attack (LF): }
This attack assumes that attackers have comprehensive control over the training process, allowing them to modify the updates in real-time. Therefore, attackers might resort to altering the training dataset instead of the update parameters. In the LF strategy, the labels of training samples are inverted: a label \( l \) is converted to \( L - l - 1 \), where \( L \) denotes the total number of classes in the classification scenario, and \( l \) is any integer from 0 to \( L-1 \).

\textbf{Combination Attack: }
The Combination Attack merges multiple adversarial techniques, such as sign-flipping, Gaussian, scaling, and label flipping attacks, to create a disruptive effect on the learning process. In this strategy, a Byzantine client may, for instance, simultaneously reverse and scale its update vector, introduce Gaussian noise, and alter training dataset labels. This multifaceted approach aims to test the model's resilience and the effectiveness of the defense mechanisms by combining the strengths of each client attack type.


\subsection{Attacks on Robust Aggregation Schemes}

In this section, we delve into the comparative analysis of \secureDL~against decentralized aggregation schemes without privacy protection, as detailed in Table \ref{tab:main}. We focus on evaluating the accuracy of all involved schemes across three diverse datasets CIFAR-10, SVHN and Fashion-MNIST. This investigation aims to highlight the effectiveness of \secureDL~in maintaining high levels of accuracy even in the presence of Byzantine attacks. 

\newcommand{\theircolor}{\cellcolor{blue!25}}
\newcommand{\ourcolor}{\cellcolor{gray!60}}

\begin{table}[htb!]
\centering
\caption{Accuracy comparison of existing decentralized robust aggregation schemes against \secureDL~under various attacks over SVHN, CIFAR-10 and Fashion-MNIST datasets. }
\label{tab:main}

\begin{subtable}{\linewidth}
\caption{CIFAR-10 dataset}
\label{tab:CIFAR10}
\centering
\setlength{\tabcolsep}{4pt} 
\resizebox{0.8\columnwidth}{!}{%
\begin{tabular}{|c|cccccc|}
\hline
\multirow{2}{*}{\# Client} &
  \multicolumn{6}{c|}{\multirow{2}{*}{N = 10}} \\
                           & \multicolumn{6}{c|}{}      \\ \hline
\# Byzantine &
  \multicolumn{1}{c|}{0} &
  \multicolumn{1}{c|}{2} &
  \multicolumn{1}{c|}{2} &
  \multicolumn{1}{c|}{2} &
  \multicolumn{1}{c|}{2} &
  4 \\ \hline
Attack type &
  \multicolumn{1}{c|}{W/O} &
  \multicolumn{1}{c|}{SF} &
  \multicolumn{1}{c|}{SA} &
  \multicolumn{1}{c|}{Noise} &
  \multicolumn{1}{c|}{LF} &
  Combi \\ \hline
Mean &
  \multicolumn{1}{c|}{\theircolor 84.21} &
  \multicolumn{1}{c|}{53.83} &
  \multicolumn{1}{c|}{13.43} &
  \multicolumn{1}{c|}{16.32} &
  \multicolumn{1}{c|}{60.27} &
  10.16 \\ \cline{1-1}
DKrum &
  \multicolumn{1}{c|}{82.20} &
  \multicolumn{1}{c|}{73.55} &
  \multicolumn{1}{c|}{74.3} &
  \multicolumn{1}{c|}{73.32} &
  \multicolumn{1}{c|}{72.95} &
  73.52 \\ \cline{1-1}
BRIDGE   & \multicolumn{1}{c|}{83.28} & \multicolumn{1}{c|}{81.86} & \multicolumn{1}{c|}{81.57} & \multicolumn{1}{c|}{81.26} & \multicolumn{1}{c|}{81.34} & 80.07 \\ \cline{1-1}
MOZI &
  \multicolumn{1}{c|}{83.36} &
  \multicolumn{1}{c|}{\theircolor 82.06} &
  \multicolumn{1}{c|}{\theircolor 81.96} &
  \multicolumn{1}{c|}{81.64} &
  \multicolumn{1}{c|}{\theircolor 82.56} &
  79.45 \\ \cline{1-1}
DMedian  & \multicolumn{1}{c|}{83.52} & \multicolumn{1}{c|}{80.36} & \multicolumn{1}{c|}{77.64} & \multicolumn{1}{c|}{\theircolor 82.64} & \multicolumn{1}{c|}{80.29} & \theircolor 80.75 \\ \cline{1-1}
\secureDL & \multicolumn{1}{c|}{\ourcolor 82.29} & \multicolumn{1}{c|}{\ourcolor 81.38} & \multicolumn{1}{c|}{\ourcolor 76.41} & \multicolumn{1}{c|}{\ourcolor 81.59} & \multicolumn{1}{c|}{\ourcolor 75.12} & \ourcolor 80.06 \\ \hline
\end{tabular}%
}
\end{subtable}

\vspace{1em} 

\begin{subtable}{\linewidth}
\centering
\caption{SVHN dataset}
\label{tab:SVHN}
\setlength{\tabcolsep}{4pt} 
\resizebox{0.8\columnwidth}{!}{%
\begin{tabular}{|c|cccccc|}
\hline
\multirow{2}{*}{\# Client} &
  \multicolumn{6}{c|}{\multirow{2}{*}{N = 10}} \\
                           & \multicolumn{6}{c|}{}      \\ \hline
\# Byzantine &
  \multicolumn{1}{c|}{0} &
  \multicolumn{1}{c|}{2} &
  \multicolumn{1}{c|}{2} &
  \multicolumn{1}{c|}{2} &
  \multicolumn{1}{c|}{2} &
  4 \\ \hline
Attack type &
  \multicolumn{1}{c|}{W/O} &
  \multicolumn{1}{c|}{SF} &
  \multicolumn{1}{c|}{SA} &
  \multicolumn{1}{c|}{Noise} &
  \multicolumn{1}{c|}{LF} &
  Combi \\ \hline
Mean &
  \multicolumn{1}{c|}{\theircolor 95.65} &
  \multicolumn{1}{c|}{21.34} &
  \multicolumn{1}{c|}{22.39} &
  \multicolumn{1}{c|}{19.99} &
  \multicolumn{1}{c|}{94.82} &
  55.35 \\ \cline{1-1}
DKrum &
  \multicolumn{1}{c|}{93.57} &
  \multicolumn{1}{c|}{93.67} &
  \multicolumn{1}{c|}{93.43} &
  \multicolumn{1}{c|}{93.71} &
  \multicolumn{1}{c|}{93.80} &
  93.84 \\ \cline{1-1}
BRIDGE &
  \multicolumn{1}{c|}{95.45} &
  \multicolumn{1}{c|}{95.34} &
  \multicolumn{1}{c|}{95.34} &
  \multicolumn{1}{c|}{95.37} &
  \multicolumn{1}{c|}{95.38} &
  94.98 \\ \cline{1-1}
MOZI &
  \multicolumn{1}{c|}{95.46} &
  \multicolumn{1}{c|}{\theircolor 95.57} &
  \multicolumn{1}{c|}{95.39} &
  \multicolumn{1}{c|}{95.40} &
  \multicolumn{1}{c|}{\theircolor 95.44} &
  94.94 \\ \cline{1-1}
DMedian &
  \multicolumn{1}{c|}{95.24} &
  \multicolumn{1}{c|}{94.96} &
  \multicolumn{1}{c|}{94.67} &
  \multicolumn{1}{c|}{95.29} &
  \multicolumn{1}{c|}{94.86} &
  94.71 \\ \cline{1-1}
\secureDL &
  \multicolumn{1}{c|}{\ourcolor 95.12} &
  \multicolumn{1}{c|}{\ourcolor 94.54} &
  \multicolumn{1}{c|}{\theircolor 95.43} &
  \multicolumn{1}{c|}{\theircolor 95.47} &
  \multicolumn{1}{c|}{\ourcolor 95.30} &
  \theircolor 95.04 \\ \hline
\end{tabular}%
}
\end{subtable}

\vspace{1em} 

\begin{subtable}{\linewidth}
\caption{Fashion-MNIST dataset}
\label{tab:FashionMNIST}
\centering
\setlength{\tabcolsep}{4pt} 
\resizebox{0.8\columnwidth}{!}{%
\begin{tabular}{|c|cccccc|}
\hline
\multirow{2}{*}{\# Client} &
  \multicolumn{6}{c|}{\multirow{2}{*}{N = 10}} \\
                           & \multicolumn{6}{c|}{}      \\ \hline
\# Byzantine &
  \multicolumn{1}{c|}{0} &
  \multicolumn{1}{c|}{2} &
  \multicolumn{1}{c|}{2} &
  \multicolumn{1}{c|}{2} &
  \multicolumn{1}{c|}{2} &
  4 \\ \hline
Attack type &
  \multicolumn{1}{c|}{W/O} &
  \multicolumn{1}{c|}{SF} &
  \multicolumn{1}{c|}{SA} &
  \multicolumn{1}{c|}{Noise} &
  \multicolumn{1}{c|}{LF} &
  Combi \\ \hline
Mean &
  \multicolumn{1}{c|}{\theircolor 93.28} &
  \multicolumn{1}{c|}{78.12} &
  \multicolumn{1}{c|}{25.23} &
  \multicolumn{1}{c|}{19.14} &
  \multicolumn{1}{c|}{87.04} &
  21.13 \\ \cline{1-1}
DKrum    & \multicolumn{1}{c|}{92.37} & \multicolumn{1}{c|}{91.76} & \multicolumn{1}{c|}{92.07} & \multicolumn{1}{c|}{91.96} & \multicolumn{1}{c|}{92.22} & 91.97 \\ \cline{1-1}
BRIDGE   & \multicolumn{1}{c|}{93.24} & \multicolumn{1}{c|}{93.04} & \multicolumn{1}{c|}{\theircolor 92.99} & \multicolumn{1}{c|}{92.90} & \multicolumn{1}{c|}{\theircolor 92.79} & 92.57 \\ \cline{1-1}
MOZI &
  \multicolumn{1}{c|}{93.22} &
  \multicolumn{1}{c|}{\theircolor 93.08} &
  \multicolumn{1}{c|}{92.98} &
  \multicolumn{1}{c|}{\theircolor 93.19} &
  \multicolumn{1}{c|}{89.62} &
  92.31 \\ \cline{1-1}
DMedian  & \multicolumn{1}{c|}{93.02} & \multicolumn{1}{c|}{92.46} & \multicolumn{1}{c|}{92.73} & \multicolumn{1}{c|}{92.89} & \multicolumn{1}{c|}{92.49} & 92.71 \\ \cline{1-1}
\secureDL & \multicolumn{1}{c|}{\ourcolor 93.23} & \multicolumn{1}{c|}{\ourcolor 92.86} & \multicolumn{1}{c|}{\ourcolor 92.25} & \multicolumn{1}{c|}{\ourcolor 92.99} & \multicolumn{1}{c|}{\ourcolor 90.30} & \theircolor 92.93 \\ \hline
\end{tabular}%
}
\end{subtable}
\end{table}

In the CIFAR-10 dataset (Table \ref{tab:CIFAR10}), the \textit{mean} aggregation method (first row), which does not provide any defense against Byzantine clients, clearly reveals its vulnerability to various Byzantine attacks. The accuracy drops to 13.43\% under scaling attacks (SA) and 16.32\% under Gaussian noise, significantly highlighting its vulnerability. In contrast, \secureDL~demonstrates robust performance across the board, maintaining an accuracy of 76.41\% against SA and, impressively, 81.59\% against noise, vastly outperforming the unprotected \textit{mean} method. It consistently demonstrates competitive resilience when comparing \secureDL's performance to other robust methods. Against SA, \secureDL's 76.41\% is superior to DKrum's 74.3\% and closely trails behind BRIDGE's 81.57\%, illustrating its effective mitigation strategies. Moreover, \secureDL's 81.59\% accuracy under Gaussian noise surpasses DKrum's 73.32\% and nearly matches BRIDGE at 81.26\%. This evidence positions \secureDL~as a strong contender in ensuring model accuracy while providing resistance against a broad spectrum of Byzantine attacks. 

In the SVHN dataset (Table \ref{tab:SVHN}), the \textit{mean} aggregation method again demonstrates its weaknesses, with accuracy falling to 21.34\% under sign-flipping attacks and dropping to 19.99\% with Gaussian noise, highlighting its struggles against complex attacks. In contrast, \secureDL~stands out for its strength, achieving a solid 94.54\% accuracy against sign-flipping and improving it to 95.47\% against noise, significantly outperforming the \textit{mean} method. These results not only show \secureDL's ability to fend off attacks but also place it in close competition with other top methods such as MOZI and BRIDGE. Its effectiveness is especially noticeable in combined attack scenarios, where it maintains a 95.04\% accuracy, showing its comprehensive approach to defense. 

Lastly, in the Fashion-MNIST dataset (Table \ref{tab:FashionMNIST}), the impact of Byzantine attacks on \textit{mean} aggregation becomes very apparent, with a drastic accuracy fall to 25.23\% under scaling attacks (SA) and a further dip to 19.14\% when faced with Gaussian noise. \secureDL, however, emerges as a formidable contender, significantly mitigating the effects of these disruptions. It maintains an impressive 92.25\% accuracy against SA and 92.99\% accuracy against noise, demonstrating its superior defense capabilities. This performance places \secureDL~in a competitive stance with BRIDGE. The results show the advantage of \secureDL~ over DMedian, especially in the case of multiple attack types (i.e., combination attack  ), where \secureDL~achieves a remarkable 92.93\% accuracy.

\begin{table*}
\centering
\resizebox{0.8\textwidth}{!}{%
\begin{tabular}{|c|cccccc|cccccc|}
\hline
Measurement & \multicolumn{6}{c|}{Time Overhead (seconds) on CIFAR-10} & \multicolumn{6}{c|}{Time Overhead (seconds) on MNIST} \\ \hline
\# Clients &
  \multicolumn{1}{c|}{3} &
  \multicolumn{1}{c|}{5} &
  \multicolumn{1}{c|}{8} &
  \multicolumn{1}{c|}{10} &
  \multicolumn{1}{c|}{20} &
  30 &
  \multicolumn{1}{c|}{3} &
  \multicolumn{1}{c|}{5} &
  \multicolumn{1}{c|}{8} &
  \multicolumn{1}{c|}{10} &
  \multicolumn{1}{c|}{20} &
  30 \\ \hline
Cosine Similarity &
  \multicolumn{1}{c|}{0.09} &
  \multicolumn{1}{c|}{0.20} &
  \multicolumn{1}{c|}{0.50} &
  \multicolumn{1}{c|}{0.73} &
  \multicolumn{1}{c|}{2.91} &
  5.63 &
  \multicolumn{1}{c|}{0.035} &
  \multicolumn{1}{c|}{0.089} &
  \multicolumn{1}{c|}{0.22} &
  \multicolumn{1}{c|}{0.36} &
  \multicolumn{1}{c|}{1.33} &
  3.20 \\
Secure Comparison &
  \multicolumn{1}{c|}{0.013} &
  \multicolumn{1}{c|}{0.039} &
  \multicolumn{1}{c|}{0.089} &
  \multicolumn{1}{c|}{0.147} &
  \multicolumn{1}{c|}{0.56} &
  1.11 &
  \multicolumn{1}{c|}{0.013} &
  \multicolumn{1}{c|}{0.035} &
  \multicolumn{1}{c|}{0.096} &
  \multicolumn{1}{c|}{0.14} &
  \multicolumn{1}{c|}{0.55} &
  1.06 \\
L2-Normalization &
  \multicolumn{1}{c|}{0.2} &
  \multicolumn{1}{c|}{0.41} &
  \multicolumn{1}{c|}{0.91} &
  \multicolumn{1}{c|}{1.28} &
  \multicolumn{1}{c|}{4.32} &
  9.13 &
  \multicolumn{1}{c|}{0.15} &
  \multicolumn{1}{c|}{0.33} &
  \multicolumn{1}{c|}{0.69} &
  \multicolumn{1}{c|}{0.99} &
  \multicolumn{1}{c|}{3.28} &
  7.34 \\ \hline
\end{tabular}%
}
\caption{Overhead measurements for protocol functions across client counts, showcasing the CNN architecture over CIFAR-10 and the MLP over MNIST.}
\label{tab:overhead}
\end{table*}

\subsection{Impact of the fraction of Byzantine clients}

\begin{figure}[tb!]
    \centering
    \begin{minipage}{0.495\linewidth}
        \centering
        \includegraphics[width=\linewidth]{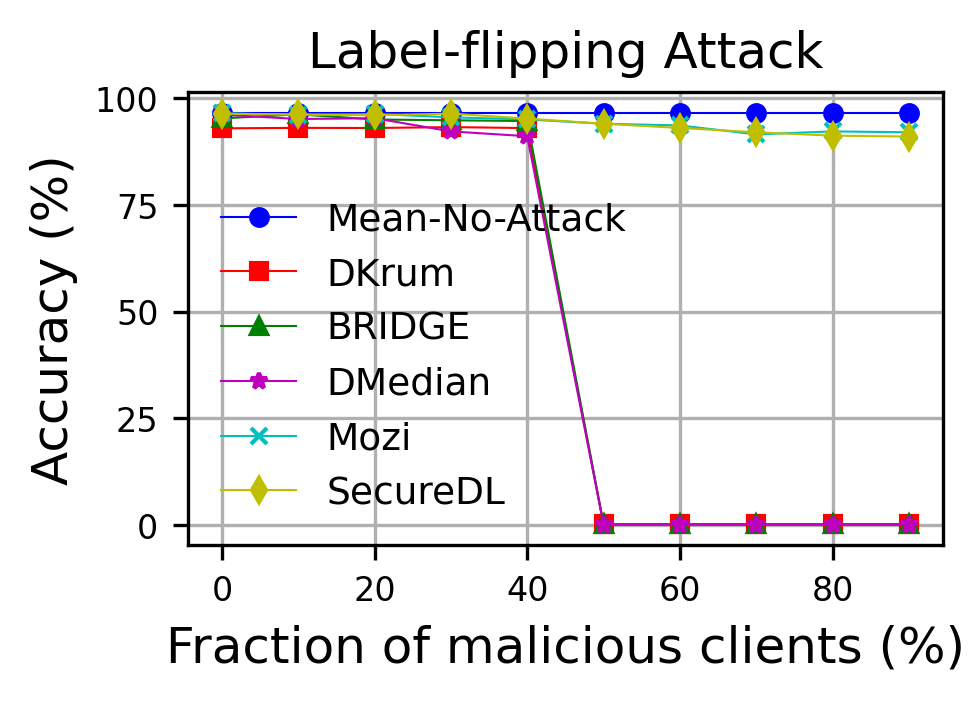}
    \end{minipage}\hfill
    \begin{minipage}{0.495\linewidth}
        \centering
        \includegraphics[width=\linewidth]{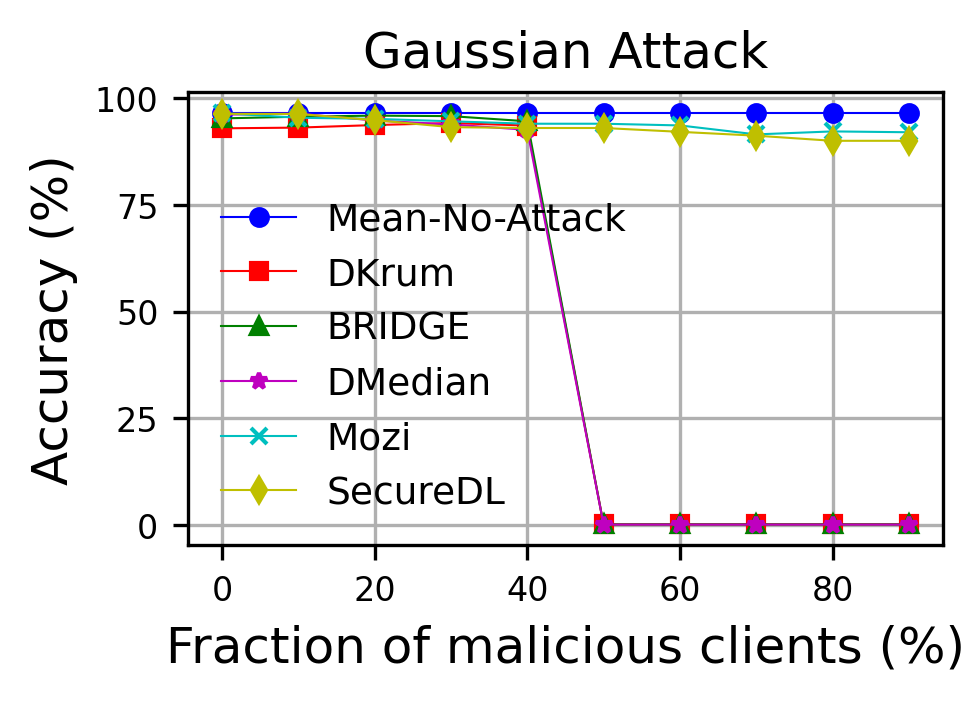}
    \end{minipage}
    
    \begin{minipage}{0.495\linewidth}
        \centering
        \includegraphics[width=\linewidth]{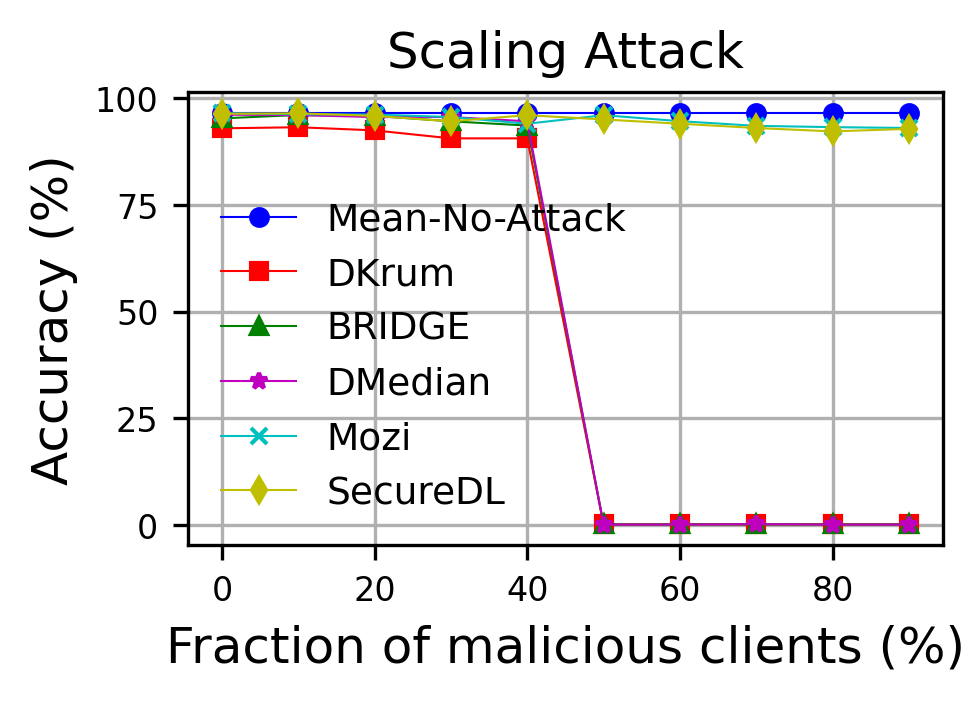}
    \end{minipage}\hfill
    \begin{minipage}{0.495\linewidth}
        \centering
        \includegraphics[width=\linewidth]{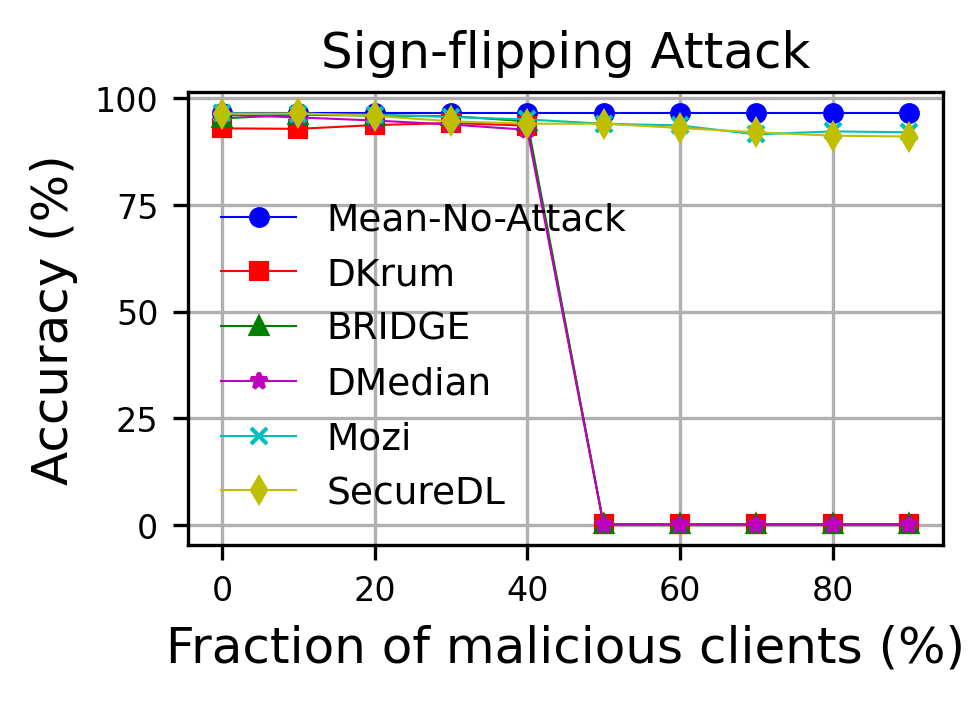}
    \end{minipage}

    \caption{Implications of varying the fraction of malicious clients in different attack scenarios on MNIST dataset. }
    \label{fractionbyz}
\end{figure}

Figure \ref{fractionbyz} illustrates the performance of various DL aggregation schemes in response to different attacks on the MNIST dataset, tracking the change in accuracy as the proportion of malicious clients varies between 0 to 80\%.

In the context of Noise and Label-flipping Attacks, the performance of Byzantine robust aggregation rules varies significantly. DKrum, reflecting the limitations of Krum's scoring mechanism, shows a notable decrease in accuracy, dropping to around 86\% at a less than 50\% malicious client level for both attack types. This indicates its vulnerability in scenarios where the proportion of malicious clients approaches its tolerance limit. BRIDGE, employing Trimmed-Mean, demonstrates a similar trend with its accuracy falling to 88\%, indicative of its effective yet limited capacity to manage Byzantine clients, aligning with its theoretical tolerance of nearly 50\% malicious clients. DMedian, influenced by the Median method's approach, maintains relatively higher accuracy (around 92\%) but begins to decline as the number of Byzantine clients approaches \( \frac{n}{2} \). Mozi, utilizing a combination of distance-based and performance-based filtering, exhibits commendable resilience, maintaining over 94\% accuracy in both attack scenarios. SecureDL, our proposed method, stands out in these settings, closely matching or slightly surpassing Mozi's robustness, particularly with 80\% of clients being Byzantine, in decentralized networks with a significant presence of Byzantine actors.


Under Sign-flipping and Scaling Attacks, we observe that DKrum's accuracy, echoing the limitations of its core selection process, declines to 86\% at a less than 50\% malicious client fraction in both attacks, underscoring its struggle in scenarios with a high concentration of Byzantine clients. The coordinate-wise strategies like BRIDGE and DMedian, while demonstrating better adaptability than DKrum, experience accuracy drops to 88\% and 92\%, respectively. This indicates their challenges in maintaining performance integrity when approaching their respective Byzantine tolerance limits. Mozi's method ensures a more stable performance, maintaining accuracy above 94\% in both attack types, thereby demonstrating its efficacy in filtering benign contributions. \secureDL, however, emerges as a notably robust method in these scenarios. Its consistent high accuracy, often slightly ahead of Mozi, particularly under scaling attacks, reinforces its suitability for decentralized learning models facing diverse Byzantine challenges. This underscores \secureDL's potential as a versatile and resilient choice in environments with a broad range of Byzantine threats.


\subsection{Computational and Communication Complexity}

In the \secureDL~protocol within a multi-party computation framework, the computational complexity is shaped by several vital steps involving all participating clients. Initially, each of the \( n \) clients engage in additive secret sharing, computing and distributing secret shares of their inputs to the others, resulting in \( n \times (n - 1) \) computational and communication tasks and yielding a complexity of \( O(n^2) \). This is followed by the computation of cosine similarities, where clients jointly compute similarities between it's own model update and every other client's update. This step involves \( n^2 \) computations, with each of the \( n \) clients jointly calculating similarities with every other client's update. This process, necessitated by security requirements for separate computations despite identical calculations between client pairs, contributes an overall complexity of \( O(n^2) \). The protocol then involves secure comparisons of these cosine similarities against a threshold maintaining the \( O(n^2) \) scaling. Finally, normalization, dependent on the number of vectors accepted based on the threshold, involves all clients and likely scales linearly with the number of clients, contributing further to the \( O(n^2) \) complexity in the worst case. Consequently, the cumulative complexity of the \secureDL~protocol is determined to be \( O(n^2) \), reflecting the quadratic scaling in computational demands with the increase in the number of participating clients.

The overall communication complexity is influenced by multiple computational steps, each with its own impact on the total communication requirements. Initially, clients engage in secret sharing with linear complexity \( O(n) \), where \( n \) is the number of clients. This is followed by the cosine similarity computation and L2 normalization, each involving \( k \) Beaver triple multiplications per pairwise computation among the \( n \) clients, resulting in a significant communication load of \( O(kn^3) \) for each step. Additionally, the protocol employs a secure comparison operation, which, based on the Rabbit protocol's complexity of \( \mathcal{O}(\ell \log \ell) \) and being performed \( n^2 \) times (once for each pair of clients), contributes significantly to the communication complexity. Assuming that the bit length \( \ell \) of the integers involved is constant and not dependent on \( n \), this step adds \( O(n^2 \ell \log \ell) \) to the overall complexity. The final step, normalization, similar to cosine similarity, also involves \( k \) Beaver triple multiplications for each of the \( n^2 \) computations, further contributing \( O(kn^3) \). Therefore, when combining these complexities, the overall communication complexity of the \secureDL~protocol is dominated by the cosine similarity, L2 normalization, and secure comparison steps, culminating in a total complexity of approximately \( O(kn^3)\). This underlines the protocol’s considerable communication demands, scaling cubically with the number of clients due to the Beaver triple operations.

\subsection{Efficiency of \secureDL}

The detailed examination of \secureDL's computational efficiency, as outlined in Table~\ref{tab:overhead}, gives indications on the time overhead for critical protocol functions—Cosine Similarity, Secure Comparison, and L2-Normalization—across a broad spectrum of client configurations, from 3 to 30, within CIFAR-10 and MNIST datasets.
Here, we omit the SVHN and Fashion MNIST datasets as they use the same model architectures with the other ones.
In a nutshell, the results show a consistent trend of increasing computational overhead with the increasing number of clients in \secureDL. For instance, in the CIFAR-10 dataset, the execution time for Cosine Similarity dramatically increases from 0.09 seconds with just three clients to 5.63 seconds with 30 clients, marking a substantial rise. L2-Normalization time exhibits a similar pattern, escalating from 0.2 seconds for three clients to 9.13 seconds for thirty clients, indicating the growing computational demands as the network size expands. This upward trend in processing time is consistently observed, albeit at lower absolute times, in the MNIST dataset, where the overhead for Cosine Similarity grows from 0.035 seconds to 3.20 seconds, and for L2-Normalization, from 0.15 seconds to 7.34 seconds as the client count increases from 3 to 30.

Crucially, the MNIST dataset evaluations were conducted over a simpler two-layer MLP network. In contrast, the CIFAR-10 assessments utilized a more complex model described in Table \ref{model_arc}, clarifying the observed discrepancies in time overheads between them. This distinction highlights the inherent differences in dataset complexity and model architecture and emphasizes \secureDL's flexibility and effectiveness across varying computational landscapes. By accommodating different network complexities, \secureDL~demonstrates a comprehensive capability to secure decentralized learning environments, effectively balancing the dual demands of robust security measures and computational efficiency.

\section{Conclusion}

In this paper, we present a decentralized machine learning approach called \secureDL~designed to enhance security and privacy in the face of Byzantine attacks. By shifting the focus from traditional client update monitoring methods to a more collaborative and secure approach, \secureDL~makes the decentralized learning framework more resilient against Byzantine and semi-honest entities. Its innovative aggregation mechanism, which rigorously scrutinizes each model update, ensures robust defense against adversarial actions while preserving the privacy of individual clients. 

The empirical evaluations of \secureDL~across various datasets demonstrates its effectiveness in providing a resilient defense mechanism against a spectrum of attacks, pushing the envelope for decentralized learning. 
Furthermore, the performance analysis of \secureDL~ sheds light to the impact of the privacy-preserving changes to the overall training process while underlining the efficiency of the proposed method. Our work contributes a vital piece to the puzzle of protecting collaborative machine learning environments by propelling the field forward with regards to critical security challenges and lays the groundwork for future innovations.

\bibliographystyle{plain}
\bibliography{sample-base}

\appendix

\section{Privacy Analysis}\label{sec:privacyAnalysis}

\begin{proof}[Proof of Theorem \ref{securedlproof}]

The functionality of the protocol \(\pi\), denoted as \(f(x, y):=w\),  is deterministic and ensures correctness. This correctness is verified by computing the sum of all outputs from \(n\) parties and evaluating this sum modulo $\mathbb{Z}_{2^k}$, specifically, \((w_1 + w_2 + \ldots + w_n\) mod $\mathbb{Z}_{2^k}$). The protocol specifies for each party \(P_i\) a corresponding component of \(f(x, y)\), \(f_i(x, y):=w_i\), where \(w_i\) is the output computed by party \(P_i\).

Define the view of party $P_i$ during the execution of protocol $\pi$ as:
\[
\operatorname{view}_i^\pi(x, y) := (x_i, y_i, r_i, m_1, m_2, \ldots, m_t) \in \mathbb{Z}_{2^k},
\]

where $r_i$ represents the outcome of the $P_i$ internal coin tosses and $m_1, m_2, \ldots, m_t$ denote the received messages. 
\\
To prove that \(\pi\) privately computes $f$, we must show that there exists a probabilistic polynomial-time simulator $S$ for every $I$ such that:
\[
\{S(x_I, f_I(x, y))\}_{x, y} \stackrel{\text{perf}}{\equiv} \{\operatorname{view}_I^\pi(x, y)\}_{x, y}.
\]

The SecureDL algorithm involves three key operations on secret-shared data: Cosine similarity computation, comparison over secret shared data and L2normalization.

The actual view of \(P_I\), for the computation is as follows:

\begin{align*}
\operatorname{view}_I^\pi(\vec{x}) = (&\vec{x}, m^{cosine}_1, m^{cosine}_2, \ldots, m^{cosine}_I, m^{comp}_1,\\
& m^{comp}_2,\ldots, m^{comp}_n,m^{l2norm}_1, m^{l2norm}_2, \ldots,\\
&m^{l2norm}_I)
\end{align*}

where \(\vec{x}\) denotes the vectors of additive shares of the inputs held by each honest parties, \(m^{cosine}_i\) are the messages exchanged during the cosine similarity computation,  \(m^{comp}_i\) are the messages exchanged for the comparison operation and \(m^{l2norm}_i\) are the messages exchanged for the L2-Normalization computation.

Given the additive sharing scheme, all elements involved in \(\operatorname{view}_I^\pi(\vec{x}\) are independently and uniformly distributed within \(\mathbb{Z}_{2^k}\), ensuring that from \(P_I\)'s perspective, all components are indistinguishable from random. This setup, based on the privacy-preserving properties of cosine similarity, secure comparison and L2Normalization operations, allows \(S_I\) to simulate a plausible set of interactions that \(P_I\) (or any group of honest parties) would observe, without revealing any private information. The simulator \(S_I\) thus constructs a simulated view that includes simulated messages for each step of the norm computation, ensuring these simulated components are indistinguishable from those in a real protocol execution.

We construct a simulator, denoted \(S_I\), for simulating the collective view of honest parties within the execution of SecureDL protocol. The role of \(S_I\) is to create a plausible simulation of the protocol execution as it would appear to the honest parties, focusing on the data exchanges and computations that occur without assuming control over the adversary's direct observations.

In the case of cosine similarity computation, \(S_I\) simulates the messages exchanged during this step, such as \(\tilde{m}^{cosine}_1, \tilde{m}^{cosine}_2, \ldots, \tilde{m}^{cosine}_n\). This simulation is consistent with the privacy assurances specified in Theorem \ref{cosine}, making these simulated messages indistinguishable from random values in \(\mathbb{Z}_{2^k}\) to any observer, thereby maintaining the integrity of the privacy guarantees.

During the comparison, \(S_I\) constructs messages such as \(\tilde{m}^{comp}_1, \tilde{m}^{comp}_2, \ldots, \tilde{m}^{comp}_n\), simulating the inter-party communications. SecureDL utilized Rabbit secure comparison protocol \cite{makri2021rabbit}, which has perfect security within \(\mathbb{Z}_{2^k}\). Therefore, to any observer, these simulated messages are indistinguishable from random values within \(\mathbb{Z}_{2^k}\), thereby upholding the integrity of the privacy guarantees.

Finally, during the L2-normalization computation, \(S_I\) constructs messages such as \(\tilde{m}^{l2norm}_1, \tilde{m}^{l2norm}_2, \ldots, \tilde{m}^{l2norm}_n\), simulating the messages exchanged during. These simulated messages are crafted to be statistically indistinguishable from those in a genuine protocol run, aligning with the privacy-preserving methodology of the square root function described in Theorem \ref{l2norm}.

The simulated view \(S_I\) provides for the protocol includes the inputs of all such parties, denoted as \(x_1, x_2, \ldots, x_m\), where \(m\) is the number of honest parties. Thus, the complete simulator view is represented as:

\begin{align*}
S_I(\bar{x}, f_I(\bar{x})) = (&\tilde{\vec{x}}, \tilde{m}^{cosine}_1, \tilde{m}^{cosine}_2, \ldots, \tilde{m}^{cosine}_n, \tilde{m}^{comp}_1\\
&, \tilde{m}^{comp}_2, \ldots, \tilde{m}^{comp}_n\, \tilde{m}^{l2norm}_1, \tilde{m}^{l2norm}_2, \ldots,\\
& \tilde{m}^{l2norm}_n)
\end{align*}

Every component of \( S_I \) is independent and uniformly random within \(\mathbb{Z}_{2^k}\)). As \(S_I\) is designed to select values that are not only independently and uniformly random but also drawn from the same distribution as those in \(\operatorname{view}_I^\pi\), these components become indistinguishable. This indistinguishability, ensures that they resemble random values, thus reinforcing the privacy aspect of the protocol.

\end{proof}

\begin{theorem}[privacy w.r.t semi-honest behavior]
\label{beavmulti}

Consider a secure multi-party computation protocol \(\pi\) designed to compute a function \(f(x, y)\) utilizing the Beaver multiplication approach (detailed in Section \ref{mul_sec}). This protocol involves \(n\) parties, denoted as \(P_1, P_2, \ldots, P_n\), with a subset \(I \subset [n]\) consisting of semi-honest parties where at least one party is honest, while the rest are semi-honest. The protocol \(\pi\) is considered to privately compute the Beaver multiplication \(f(x, y)\) if a probabilistic polynomial-time simulator \(S_i\) can be constructed such that, given all potential inputs, the simulated view of \(P_i\) is statistically indistinguishable from \(P_i\)'s actual view during the execution of the protocol.
\end{theorem}

\begin{proof}
The functionality of the protocol \(\pi\), denoted as \(f(x, y):=w\),  is deterministic and ensures correctness. This correctness is verified by computing the sum of all outputs from \(n\) parties and evaluating this sum modulo $\mathbb{Z}_{2^k}$, specifically, \((w_1 + w_2 + \ldots + w_n\) mod $\mathbb{Z}_{2^k}$). The protocol specifies for each party \(P_i\) a corresponding component of \(f(x, y)\), \(f_i(x, y):=w_i\), where \(w_i\) is the output computed by party \(P_i\).

Define the view of party $P_i$ during the execution of protocol $\pi$ as:
\[
\operatorname{view}_i^\pi(x, y) := (x_i, y_i, r_i, m_1, m_2, \ldots, m_t) \in \mathbb{Z}_{2^k},
\]

where $r_i$ represents the outcome of the $P_i$ internal coin tosses and $m_1, m_2, \ldots, m_t$ denote the received messages. 
\\
To prove that \(\pi\) privately computes $f$, we must show that there exists a probabilistic polynomial-time simulator $S$ for every $I$ such that:
\[
\{S(x_I, f_I(x, y))\}_{x, y} \stackrel{\text{perf}}{\equiv} \{\operatorname{view}_I^\pi(x, y)\}_{x, y}.
\]

The actual view of $P$ in protocol $\pi$ for the multiplication function \( f(x, y) \) is :
\begin{align*}
\operatorname{view}_I^\pi(x, y) = (&x_I, y_I, a_I, b_I, c_I, \delta_I,\epsilon_I)
\end{align*}

where \(y_1, \delta_I,\) and \(\epsilon_I\) are messages received by \(P_I\) from other parties. Also, \(a_I\), \(b_I\) and \(c_I\) are the Beaver triples received by \(P_I\) for a multiplication.

It is important to note that all these elements are additively shared between the parties or results of the addition of two additively shared values. Therefore, based on the definition in Section \ref{mul_sec}, they are independent uniformly random within the \(\mathbb{Z}_{2^k}\) set. 

Furthermore, from the $\operatorname{view}_I^\pi(x, y)$, \(P_I\) can compute all other relevant quantities and this view contains all of \(P_I\)'s information.

Now, we construct a simulator, denoted \(S\), for the view of party \(P_I\) that simulates the computation of \(\delta\) and \(\epsilon\), representing the differences between simulated inputs \(\tilde{x}_I\) and \(\tilde{y}_I\), and predetermined shares \(\tilde{a}_I\) and \(\tilde{b}_I\) from Beaver triples. These simulated values are independent and chosen uniformly at random, ensuring they mirror the actual protocol's input and share distribution. The simulator then calculates \(\tilde{\delta}_I = \tilde{x}_I - \tilde{a}_I\) and \(\tilde{\epsilon}_I = \tilde{y}_I - \tilde{b}_I\), mirroring the protocol's approach to determining \(\delta\) and \(\epsilon\).

With \({\delta}\) and \({\epsilon}\) computed, \(S\) proceeds to simulate the computation of the share of the product, \(\tilde{w}_I\), using the formula \(\tilde{w}_I = \tilde{c}_I + {\epsilon} \cdot \tilde{b}_I + {\delta} \cdot \tilde{a}_I + {\epsilon} \cdot {\delta}\). This step effectively mimics the actual protocol's computation, incorporating the adjustment based on the computed differences (\(\delta\) and \(\epsilon\)) and the correction factor (\(\tilde{\epsilon_I} \cdot \tilde{\delta_I}\)), utilizing the simulated shares of \(a\), \(b\), and \(c\) from the Beaver triples. 

We define simulator $S$ view for the multiplication as

\begin{align*}
S(x_I, f_I(x, y)) = (&x, \tilde{x_1}, \tilde{y_1}, \tilde{a_I}, \tilde{b_I}, \tilde{c_I},\tilde{\delta_I}\tilde{\epsilon_I})
\end{align*}

where each component of $S$ is independent uniformly random in \(\mathbb{Z}_{2^k}\). Since \(S\) selects and computes values that are independently and uniformly random, and from the same distribution as those in \(\operatorname{view}^\pi\), they are indistinguishable. As a result, the statistical distance between the simulated view \(S(x_I, f_I(x, y))\) and the actual view \(\operatorname{view}_I^\pi(x, y)\) is effectively zero for all \(x, y\). Consequently, all components, except for \(x\), are indistinguishable from random values, thereby affirming the protocol's privacy. \\

\end{proof}

\begin{theorem}[privacy w.r.t semi-honest behavior]
\label{invproof}
Consider a secure multi-party computation protocol \(\pi\) designed to compute a function \(f(x)\) utilizing the inverse function (detailed in Section \ref{inv}). This protocol involves \(n\) parties, denoted as \(P_1, P_2, \ldots, P_n\), with a subset \(I \subset [n]\) consisting of semi-honest parties where at least one party is honest, while the rest are semi-honest. The protocol \(\pi\) is considered to privately compute the inverse function \(f(x)\) if a probabilistic polynomial-time simulator \(S_i\) can be constructed such that, given all potential inputs, the simulated view of \(P_i\) is statistically indistinguishable from \(P_i\)'s actual view during the execution of the protocol.
\end{theorem}

\begin{proof}
The protocol \(\pi\) ensures deterministic correctness by summing the outputs of \(n\) parties modulo \(\mathbb{Z}_{2^k}\). Each party \(P_i\) contributes a specific output \(w_i\) to the function \(f(x, y)\). The privacy of \(\pi\) is shown by proving that for any inputs \(x\) and \(y\), a simulator can generate a view for party \(P\) that is indistinguishable from its actual view during the protocol's execution, encompassing inputs, internal randomness, and received messages.

The inversion function consists of three main operations on secret-shared inputs: addition, multiplication by a public constant, and Beaver multiplication. In contrast to Beaver multiplication, both the addition and the multiplication by a public constant operations are executed locally on the secret-shared data.

The inverse function begins with a ``for'' loop, within which there are two Beaver multiplications; the output of the first multiplication serves as the input for the second Beaver multiplication. Therefore, we can define the operations inside this loop as \(f(x, y)\) and represent this function in the following form:

\begin{align}
f'(x,y) = x \times y = w \nonumber\\
f''(w, x) = w \times x \nonumber
\end{align}

Here, \( w \) is the output of the first multiplication \( f'(x, y) \), and the output of \( f''(w, x) \) is equal to \( f(x, y) \).

Based on Theorem \ref{beavmulti}, protocol \(\pi\) can compute each of the multiplication privately. As long as \( w1 \), \( w2 \), \ldots and \(w_n\), the shares of the Beaver multiplications output, are not revealed, they remain uniformly random and independent. However, since the output of the first multiplication is the input of the second multiplication, we need to show that the first multiplication \( f'(x, y) \) remains private when the second multiplication \( f''(w, x) \) is executed. In other words, the information received by \( P \) during the second multiplication should not reveal any details about the inputs \( x \) and \( y \) of the first multiplication. \\

Therefore, we need to define the actual view of \(P_I\) within protocol \(\pi\) specifically for the second multiplication:

\begin{align*}
\operatorname{view}_I^\pi(w, x) = &\ (x_I, y_I, w_I, a_I, b_I, c_I, a_I', b_I', c_I', \\
&\ \delta_2, \delta_3, \ldots, \delta_n, \epsilon_2, \epsilon_3, \ldots, \epsilon_n, \\
&\ \delta_2', \delta_3', \ldots, \delta_n', \epsilon_2', \epsilon_3', \ldots, \epsilon_n')
\end{align*}

where $y_1$ is the share of initial guess, \(\delta_2, \delta_3, \ldots, \delta_n,\) and \(\epsilon_2, \epsilon_3, \ldots, \epsilon_n\) are messages received by \(P\) from other parties for the second multiplication, and \(\delta_2', \delta_3', \ldots, \delta_n', \epsilon_2', \epsilon_3', \ldots, \epsilon_n'\) are messages received by \(P\) from other parties for the first multiplication. Also, \(a_I\), \(b_I\) and \(c_I\) and \(a_I'\), \(b_I'\) and \(c_1'\) are the Beaver triples received by \(P\) for the multiplications and \(w_1\) is the result of first multiplication.

All these elements, except for \(x\), are additively shared between the parties or results of the addition of two additively shared values. Therefore, based on the definition in \ref{mul_sec}, they are independent uniformly random within the \(\mathbb{Z}_{2^k}\) set. \\

Furthermore, from the $\operatorname{view}_I^\pi(w, x)$, \(P\) can compute all other relevant quantities and this view contains all of \(P\)'s information.

Now, we construct a simulator, denoted \(S\), for the view of party \(P\) that simulates the computation of \(\delta\) and \(\epsilon\), representing the differences between simulated inputs \(\tilde{x}_I\) and \(\tilde{y}_I\), and predetermined shares \(\tilde{a}_I\) and \(\tilde{b}_I\) from Beaver triples. These simulated values are independent and chosen uniformly at random, ensuring they mirror the actual protocol's input and share distribution. The simulator then calculates \(\tilde{\delta} = \tilde{x} - \tilde{a}_I\) and \(\tilde{\epsilon} = \tilde{y} - \tilde{b}_I\), mirroring the protocol's approach to determining \(\delta\) and \(\epsilon\).

With \({\delta}\) and \({\epsilon}\) computed, \(S\) proceeds to simulate the computation of the share of the product, \(\tilde{w}_I\), using the formula \(\tilde{w}_I = \tilde{c}_I + {\epsilon} \cdot \tilde{b}_I + {\delta} \cdot \tilde{a}_I + {\epsilon} \cdot {\delta}\). This step effectively mimics the actual protocol's computation, incorporating the adjustment based on the computed differences (\(\delta\) and \(\epsilon\)) and the correction factor (\(\tilde{\epsilon} \cdot \tilde{\delta}\)), utilizing the simulated shares of \(c\), \(a\), and \(b\) from the Beaver triples.

Now, we define simulator $S$ as
\begin{align*}
S(x_I, f_I(w, x)) = (&\tilde{x}_I, \tilde{y}_I, \tilde{a}_I, \tilde{b}_I, \tilde{c}_I, \tilde{a}_I', \tilde{b}_I', \tilde{c}_I', \tilde{\delta}_2\\
&, \tilde{\delta}_3, \ldots, \tilde{\delta}_n, \tilde{\delta}_2', \tilde{\delta}_3', \ldots, \tilde{\delta}_n', \tilde{\epsilon}_2, \tilde{\epsilon}_3\\
&, \ldots, \tilde{\epsilon}_n, \tilde{\epsilon}_2', \tilde{\epsilon}_3', \ldots, \tilde{\epsilon}_n')
\end{align*}

Unlike the first beaver multiplication, the simulator view for the second multiplication consists of all messages \( P \) receives during the second multiplication, plus what it could see from the first multiplication. 

As shown, all components of \( S \) are independent and uniformly random in \(\mathbb{Z}_{2^k}\), except for \( x \), which is \( P \)'s private input. Since \( S \) selects and computed values that are independently and uniformly random, and from the same distribution as those in \(\operatorname{view}^\pi\), they are indistinguishable. As a result, the statistical distance between the simulated view \( S(x, f''_I(w, x)) \) and the actual view \(\operatorname{view}_I^\pi(w, x) \) is effectively zero for all \( w, x \). Consequently, all components, except for \( x \), are indistinguishable from random values, thereby affirming the protocol's privacy.

Moreover, since for each multiplication, only a new set of Beaver triples is used, the values \(\tilde{\gamma_n}\), \(\tilde{\delta_n}\) are different from \(\tilde{\gamma_n'}\), \(\tilde{\delta_n'}\) in the second multiplication. Therefore, all these intermediate value remain indistinguishable from random values.

Before discussing the actual and simulated views of the entire inversion function, we first need to demonstrate that the local operations on share values in inverse function do not reveal any information about the output share of the second multiplication. As shown in this algorithm, each client subtract its output share of the second multiplication (which we call \(w'\)) from twice of one of the input share of the multiplication, \(x\). Therefore, we need to demonstrate that simulator \(S\) can mimic the actual protocol's computation:
\[
[\![ w'' ]\!] = 2 \times [\![ x ]\!] - [\![ w' ]\!] 
\]
To ensure the indistinguishability of the simulated computation from the actual protocol's operation $2 \times [\![ x ]\!]_i - [\![ w' ]\!]_i$, the simulator $S$ selects uniformly distributed and independent random variables, $[\![\tilde{x} ]\!]$ and $[\![\tilde{w}' ]\!]$. These variables are crucial for mimicking the actual computation without revealing any sensitive information about the output of the second multiplication, $w'$, or the input, $x$. By calculating $[\![\tilde{w}'']\!] = 2 \times [\![\tilde{x}]\!] - [\![\tilde{w}']\!]$ with these selected variables, $S$ manages to simulate the operation in a way that the outcome mirrors the distribution and independence characteristics of the actual computation. 

The above steps represent a single execution of the loop inside the inversion function. To show the complete simulator view for this function, we need to demonstrate the simulator view when these steps are repeated \(k\) times:

\begin{align*}
S(x_I, f_I(x)) = (&\tilde{x}_I, \tilde{y}_I, \tilde{a}_I^{(1)}, \tilde{b}_I^{(1)}, \tilde{c}_I^{(1)}, \tilde{a}_I'^{(1)}, \tilde{b}_I'^{(1)}, \tilde{c}_I'^{(1)}, \tilde{\delta}_2^{(1)}, \\
&\tilde{\delta}_3^{(1)}, \ldots, \tilde{\delta}_n^{(1)}, \tilde{\delta}_2'^{(1)}, \tilde{\delta}_3'^{(1)}, \ldots, \tilde{\delta}_n'^{(1)}, \tilde{\epsilon}_2^{(1)}, \tilde{\epsilon}_3^{(1)},\\
& \ldots, \tilde{\epsilon}_n^{(1)}, \tilde{\epsilon}_2'^{(1)}, \tilde{\epsilon}_3'^{(1)}, \ldots, \tilde{\epsilon}_n'^{(1)}, \ldots, \tilde{a}_1^{(k)}, \tilde{b}_1^{(k)},\\
& \tilde{c}_1^{(k)}, \tilde{a}_1'^{(k)}, \tilde{b}_1'^{(k)}, \tilde{c}_1'^{(k)}, \tilde{\delta}_2^{(k)}, \tilde{\delta}_3^{(k)}, \ldots, \tilde{\delta}_n^{(k)}, \\
&\tilde{\delta}_2'^{(k)}, \tilde{\delta}_3'^{(k)}, \ldots, \tilde{\delta}_n'^{(k)}, \tilde{\epsilon}_2^{(k)}, \tilde{\epsilon}_3^{(k)}, \ldots, \tilde{\epsilon}_n^{(k)}, \tilde{\epsilon}_2'^{(k)}, \\
&\tilde{\epsilon}_3'^{(k)}, \ldots, \tilde{\epsilon}_n'^{(k)}, \tilde{w}^{(1)}, \tilde{w}'^{(1)}, \tilde{w}''^{(1)}, \ldots, \tilde{w}^{(k)},\\
& \tilde{w}'^{(k)}, \tilde{w}''^{(k)}
 )
\end{align*}

Every component of \( S \) is independent and uniformly random within \(\mathbb{Z}_{2^k}\), with the exception of \( x \), which represents \( P \)'s private input. As \( S \) is designed to select values that are not only independently and uniformly random but also drawn from the same distribution as those in \(\operatorname{view}_I^\pi\), these components become indistinguishable. This indistinguishability, applicable to all components other than \( x \), ensures that they resemble random values, thus reinforcing the privacy aspect of the protocol.

\end{proof}

\begin{theorem}[privacy w.r.t semi-honest behavior]
\label{sqr_proof}

Consider a secure multi-party computation protocol \(\pi\) designed to compute a function \(f(x)\) utilizing the square root function (detailed in Section \ref{sqr}). This protocol involves \(n\) parties, denoted as \(P_1, P_2, \ldots, P_n\), with a subset \(I \subset [n]\) consisting of semi-honest parties where at least one party is honest, while the rest are semi-honest. The protocol \(\pi\) is considered to privately compute the square root function \(f(x)\) if a probabilistic polynomial-time simulator \(S_i\) can be constructed such that, given all potential inputs, the simulated view of \(P_i\) is statistically indistinguishable from \(P_i\)'s actual view during the execution of the protocol.
\end{theorem}

\begin{proof}

The protocol \(\pi\) ensures deterministic correctness by summing the outputs of \(n\) parties modulo \(\mathbb{Z}_{2^k}\). Each party \(P_i\) contributes a specific output \(w_i\) to the function \(f(x)\). The privacy of \(\pi\) is shown by proving that for any input \(x\), a simulator can generate a view for party \(P\) that is indistinguishable from its actual view during the protocol's execution, encompassing inputs, internal randomness, and received messages.

The square root function consists of four main operations on secret-shared data: addition, multiplication by a public constant, Beaver multiplication and inversion. The addition and the multiplication by a public constant operations are executed locally on the secret-shared data.

The square root function begins with a 'for' loop that iterates \(k\) times, with \(k\) being fixed and known to all parties. Within this loop, the inverse of the initial guess is computed. Subsequently, the output of the inversion function is multiplied by the share of input of the square root function using Beaver multiplication. Following this, the output of this multiplication is locally added to the initial guess. Finally, the output from the previous computation is multiplied by 1/2 to yield the final approved result for the square root function.

We now define the actual view of \(P\) within protocol \(\pi\) for the square function:

\begin{align*}
\operatorname{view}_I^\pi(x) = (&x_I, m^{inv(1)}_2, m^{inv(1)}_3, \ldots, m^{inv(1)}_n, m^{mul(1)}_2,\\
&m^{mul(1)}_3, \ldots, m^{mul(1)}_n, \ldots, m^{inv(k)}_2, m^{inv(k)}_3\\
&, \ldots, m^{inv(k)}_n, m^{mul(k)}_2, m^{mul(k)}_3, \ldots, m^{mul(k)}_n)
\end{align*}

where $x_I$ is a share of the input of square function, \(m^{inv}_1, m^{inv}_2, \ldots, m^{inv}_n\)  are messages received by \(P\) from other parties for the inversion operation and \(m^{mul}_2 , m^{mul}_3, \ldots,  m^{mul}_n)\) are messages received by \(P\) from other parties for the Beaver multiplication operation in each of $k$ iterations.

All these elements, are additively shared between the parties or results of the addition of two additively shared values. Therefore, based on the definition in Algorithm \ref{ss}, they are independent uniformly random within the \(\mathbb{Z}_{2^k}\) set. \\

Furthermore, from the $\operatorname{view}_I^\pi(x)$, \(P\) can compute all other relevant quantities and this view contains all of \(P\)'s information.

Now, we construct a simulator, denoted \(S\), for the view of party \(P\) that simulates the computation of the square root function. The simulator \(S\) begins by simulating the initial setup, including the distribution of the secret-shared input $\tilde{Y_1}$ and the initial guess $\tilde{x_0}$ for the square root computation. For the inversion operation, \(S\) simulates the messages received from other parties during the execution of the inverse function \(\tilde{m}^{inv}_1, \tilde{m}^{inv}_2, \ldots, \tilde{m}^{inv}_n\) and the outcome of the inversion function \(\tilde{w}^{inv}\). According to Theorem \ref{invproof}, the inverse function can be computed privately, and the simulator can successfully simulate these exchanged messages, which are indistinguishable from random values in \(\mathbb{Z}_{2^k}\).

For the Beaver multiplication, \(S\) simulates the Beaver triples and the messages exchanged during the Beaver multiplication process \(\tilde{m}^{mul}_1, \tilde{m}^{mul}_2, \ldots, \tilde{m}^{mul}_n\). In simulating the Beaver multiplication function, \(S\) adheres to the privacy guarantees provided by Theorem \ref{beavmulti}. Therefore, the simulator can successfully simulate these exchanged messages, which are indistinguishable from random values in \(\mathbb{Z}_{2^k}\). Also, since in actual protocol, parties do not reveal the shares of the output of Beaver multiplication function \(\tilde{w}^{mul}_n\), \(S\) can simulate \(\tilde{w}^{mul}_1\) as uniformly random and independent values in \(\mathbb{Z}_{2^k}\) that makes it indistinguishable.

The final steps of the algorithm involve local computations by \(P\), including adding the result of Beaver multiplication (which is called here $[\![w']\!]$) to the current estimate of square root ( $[\![x_n]\!]$) and then multiplying by \(1/2\) to adjust for the next iteration or final result. In this step,  \(S\) simulates these computations by following the same steps, the simulator $S$ selects uniformly distributed and independent random variables, $[\![\tilde{w}' ]\!]$ and $[\![\tilde{x_n}' ]\!]$ from \(\mathbb{Z}_{2^k}\) and compute the final share value for the this round of approximation as follows:
\[
[\![\tilde{x}_{n+1}]\!] \leftarrow \frac{1}{2} \left( [\![\tilde{x}_{n}]\!] + [\!\tilde{[w']}\!]\right)
\]

These steps effectively mimic the actual protocol’s computation, and since the square root computation iterates this process \(k\) times, \(S\) repeats the simulation for each iteration, updating the simulated values according to the algorithm's steps.

Now we show the complete simulator view for square function:

\begin{align*}
S(x_I, f_I(x)) = (&x_I, \tilde{m}^{inv(1)}_2, \tilde{m}^{inv(1)}_3, \ldots, \tilde{m}^{inv(1)}_n, \tilde{m}^{inv}_2, \\
&\tilde{m}^{inv(2)}_3,\ldots, \tilde{m}^{inv(2)}_n, \ldots, \tilde{m}^{inv(k)}_2, \\
&\tilde{m}^{inv(k)}_3, \ldots, \tilde{m}^{inv(k)}_n, \tilde{m}^{mul(1)}_2, \\
&\tilde{m}^{mul(1)}_3, \ldots, \tilde{m}^{mul(1)}_n, \tilde{m}^{mul(2)}_2,\\
&\tilde{m}^{mul(2)}_3, \ldots, \tilde{m}^{mul(2)}_n, \ldots, \tilde{m}^{mul(k)}_2, \\
&\tilde{m}^{mul(k)}_3, \ldots, \tilde{m}^{mul(k)}_n, \tilde{w}^{mul(1)}_1, \\
&\tilde{w}^{mul(2)}_1, \ldots, \tilde{w}^{mul(k)}_1, \tilde{w}^{inv(1)}_1, \\
&\tilde{w}^{inv(2)}_1, \ldots,\tilde{w}^{inv(k)}_1)
\end{align*}

Every component of \( S \) is independent and uniformly random within \(\mathbb{Z}_{2^k}\) which represents \( P \)'s private input. As \( S \) is designed to select values that are not only independently and uniformly random but also drawn from the same distribution as those in \(\operatorname{view}_I^\pi\), these components become indistinguishable. This indistinguishability, applicable to all components ensures that they resemble random values, thus reinforcing the privacy aspect of the protocol.

\end{proof}

\begin{theorem}[privacy w.r.t semi-honest behavior]
\label{norm}
Consider a secure multi-party computation protocol \(\pi\) designed to compute the norm of a vector, as detailed in line 1-3 of Algorithm \ref{ComputeL2Normalization}. This protocol involves \(n\) parties, denoted as \(P_1, P_2, \ldots, P_n\), with a subset \(I \subset [n]\) consisting of semi-honest parties where at least one party is honest, while the rest are semi-honest. Each party \(P_i\) holds an additive share of a set of private inputs \(x\), where \(x\) represents the collective input vector shared among the parties. The protocol \(\pi\) is said to privately compute the norm function \(f(x)\) if there exists a probabilistic polynomial-time simulator \(S\) that can simulate a view for any subset of parties that is statistically indistinguishable from their actual views during the protocol's execution, given all potential inputs.

\end{theorem}

\begin{proof}

The protocol \(\pi\) guarantees deterministic correctness by computing the sum of the outputs from \(n\) parties modulo \(\mathbb{Z}_{2^k}\), where each party \(P_i\) contributes an output \(w_i\) corresponding to their share of the computation of \(f(x)\). To demonstrate the privacy of \(\pi\), we show that for any collective input \(x\), a simulator \(S\) can generate a simulated view for any honest party or group of honest parties that mirrors their actual experience during the protocol's execution. This includes their inputs, the internal randomness used, and the messages received from other parties.

The norm computation involves three key operations on secret-shared data: Beaver multiplication for squaring the input vector, local summation of the squared values, and finally, square root computation to derive the norm.

The actual view of \(P_I\), for the computation is as follows:

\[
\operatorname{view}_I^\pi(x) = (x_I, m^{mul}_1, m^{mul}_2, \ldots, m^{mul}_n, m^{sqrt}_1, m^{sqrt}_2, \ldots, m^{sqrt}_n),
\]

where \(x_I\) denotes the vector of additive shares of the input held by all honest parties, \(m^{mul}_i\) are the messages exchanged during the Beaver multiplication, and \(m^{sqrt}_i\) are the messages exchanged for the square root computation.

Given the additive sharing scheme, all elements involved in \(\operatorname{view}_I^\pi(\vec{x})\) are independently and uniformly distributed within \(\mathbb{Z}_{2^k}\), ensuring that from \(P\)'s perspective, all components are indistinguishable from random. This setup, based on the privacy-preserving properties of Beaver multiplication and square root operations, allows \(S\) to simulate a plausible set of interactions that \(P\) (or any group of honest parties) would observe, without revealing any private information. The simulator \(S\) thus constructs a simulated view that includes simulated messages for each step of the norm computation, ensuring these simulated components are indistinguishable from those in a real protocol execution.

We construct a simulator, denoted \(S\), for simulating the collective view of honest parties in the vector norm computation within the SecureDL protocol. The role of \(S\) is to create a plausible simulation of the protocol execution as it would appear to the honest parties, focusing on the data exchanges and computations that occur without assuming control over the adversary's direct observations.

In the case of Beaver multiplication, \(S\) simulates the Beaver triples and the messages exchanged during this phase, such as \(\tilde{m}^{mul}_1, \tilde{m}^{mul}_2, \ldots, \tilde{m}^{mul}_n\). This simulation is consistent with the privacy assurances specified in Theorem \ref{beavmulti}, making these simulated messages indistinguishable from random values in \(\mathbb{Z}_{2^k}\) to any observer, thereby maintaining the integrity of the privacy guarantees.

For local computations, including the summation of outputs, the simulator \(S\) performs:

\[
[\![\tilde{\text{sum}}]\!] \leftarrow \sum_{i=1}^{n} [\![\tilde{w}_i^{mul}]\!],
\]

imitating the aggregation process faithfully without disclosing any individual's private data.

During the square root computation, \(S\) constructs messages such as \(\tilde{m}^{sqrt}_1, \tilde{m}^{sqrt}_2, \ldots, \tilde{m}^{sqrt}_n\), simulating the inter-party communications. These simulated messages are crafted to be statistically indistinguishable from those in a genuine protocol run, aligning with the privacy-preserving methodology of the square root function described in Theorem \ref{sqr_proof}.

Considering multiple honest parties, the simulated view \(S\) provides for the protocol includes the inputs of all such parties, denoted as \(x_1, x_2, \ldots, x_m\), where \(m\) is the number of honest parties. Thus, the complete simulator view is represented as:

\[
S(x_I, f_I(x)) = (x_I, \tilde{m}^{mul}_1, \tilde{m}^{mul}_2, \ldots, \tilde{m}^{mul}_n, \tilde{m}^{sqrt}_1, \tilde{m}^{sqrt}_2, \ldots, \tilde{m}^{sqrt}_n),
\]

Each component of \(S\) is independently and uniformly random within \(\mathbb{Z}_{2^k}\), ensuring that the simulation reflects the private inputs and the exchange of messages among honest parties, all while remaining indistinguishable from a real execution to any external observer.

\end{proof}

\begin{theorem}[privacy w.r.t semi-honest behavior]
\label{cosine}
Consider a secure multi-party computation protocol \(\pi\) designed to compute a function \(f(x, y)\), utilizing the cosine similarity computation between two vectors (as detailed in Section \ref{alg:CosineSimilarityIntegratedNorm}). This protocol involves \(n\) parties, denoted as \(P_1, P_2, \ldots, P_n\), with a subset \(I \subset [n]\) consisting of semi-honest parties where at least one party is honest, while the rest may be semi-honest. Each party \(P_i\) possesses an additive share of the private inputs \(x\) and \(y\). The protocol \(\pi\) is considered to effectively compute the cosine similarity function \(f(x, y)\) in a private manner if a probabilistic polynomial-time simulator \(S\) can be constructed. For any given set of potential inputs, the simulated view of \(P_i\) should be statistically indistinguishable from \(P_i\)'s actual view during the execution of the protocol.\end{theorem}

\begin{proof}

The protocol \(\pi\) ensures deterministic correctness by summing the outputs of \(n\) parties modulo \(\mathbb{Z}_{2^k}\). Each party \(P_i\) contributes a specific output \(w_i\) to the function \(f(x, y)\). The privacy of \(\pi\) is shown by proving that for any inputs \(x\) and \(y\), a simulator can generate a view for party \(P\) that is indistinguishable from its actual view during the protocol's execution, encompassing inputs, internal randomness, and received messages.

In the first step of the cosine similarity computation function, the dot product of the inputs should be computed. In our work, a single Beaver multiplication is utilized to perform this operation. Following this, the norm of each vector is computed, and then these norm values are multiplied using Beaver triple multiplication. Subsequently, the inverse of this multiplication needs to be computed. Finally, the result of that last step is multiplied by the dot product results computed in the first step to calculate the final output of the cosine similarity function.

We now define the actual view of \(P\) within protocol \(\pi\) for the square function:

\begin{align*}
\operatorname{view}_I^\pi(x,y) = (&x_I, y_I  m^{mul_1}_1,m^{mul_1}_2, \ldots, m^{mul_1}_n, m^{norm_1}_1,\\
&m^{norm_1}_2, \ldots, m^{norm_1}_n, m^{norm_2}_1,m^{norm_2}_2, \ldots,\\
&m^{norm_2}_n, m^{mul_2}_1,m^{mul_2}_2, \ldots, m^{mul_2}_n,m^{inv}_1,\\
&m^{inv}_2, \ldots, m^{inv}_n, m^{mul_3}_1,m^{mul_3}_2, \ldots, m^{mul_3}_n)
\end{align*}

where $x_I$ and $y_I$ are the share of the inputs of norm computation function, \(m^{mul_1}_1,m^{mul_1}_2, \ldots, m^{mul_1}_n\), \(m^{mul_2}_1,m^{mul_2}_2, \ldots, m^{mul_2}_n\) and \(m^{mul_3}_1,m^{mul_3}_2, \ldots, m^{mul_3}_n\) are messages received by \(P\) from other parties for the Beaver multiplication operations and \(m^{norm_1}_1, m^{norm_1}_2, \ldots, m^{norm_1}_n\) and \(m^{norm_2}_1, m^{norm_2}_2, \ldots, m^{norm_2}_n\) are messages received by \(P\) from other parties for the norm computation operation and finally \(m^{inv}_1, m^{inv}_2, \ldots, m^{inv}_n\) are messages received by \(P\) from other parties for the inversion operation.

All these elements, are additively shared between the parties or results of the addition of two additively shared values. Therefore, based on the definition in Algorithm \ref{ss}, they are independent uniformly random within the \(\mathbb{Z}_{2^k}\) set. Furthermore, from the $\operatorname{view}_I^\pi(x,y)$, \(P\) can compute all other relevant quantities and this view contains all of \(P\)'s information. \\

Now, we construct a simulator, denoted \(S\), for the view of party \(P\) that simulates the computation of the cosine similarity computation function. The simulator \(S\) begins by simulating the initial setup, including the of the secret-shared inputs $\tilde{w_a}$ and $\tilde{w_b}$. For the Beaver multiplications, \(S\) simulates the Beaver triples and the messages exchanged during the Beaver multiplication process \(\tilde{m}^{mul_1}_1, \tilde{m}^{mul_1}_2, \ldots, \tilde{m}^{mul_1}_n\), \(\tilde{m}^{mul_2}_1, \tilde{m}^{mul_2}_2, \ldots, \tilde{m}^{mul_2}_n\) . In simulating the Beaver multiplication function, \(S\) adheres to the privacy guarantees provided by Theorem \ref{beavmulti}. Therefore, the simulator can successfully simulate these exchanged messages, which are indistinguishable from random values in \(\mathbb{Z}_{2^k}\). Also, since in actual protocol, parties do not reveal the shares of the output of Beaver multiplication functions in any steps \(\tilde{w}^{mul_1}_n\),\(\tilde{w}^{mul_2}_n\)and \(\tilde{w}^{mul_3}_n\), \(S\) can simulate them as uniformly random and independent values in \(\mathbb{Z}_{2^k}\) that makes it indistinguishable.

Next step involves vector norm computations, \(S\) simulates the messages exchanged during the computation process \(\tilde{m}^{norm_1}_1, \tilde{m}^{norm_1}_2, \ldots, \tilde{m}^{norm_1}_n\) and \(\tilde{m}^{norm_2}_1, \tilde{m}^{norm_2}_2, \ldots, \tilde{m}^{norm_2}_n\) . In simulating the norm computation function, \(S\) adheres to the privacy guarantees provided by Theorem \ref{norm}. Therefore, the simulator can successfully simulate these exchanged messages, which are indistinguishable from random values in \(\mathbb{Z}_{2^k}\). Also, since in actual protocol, parties do not reveal the shares of the output of norm functions in any steps \(\tilde{w}^{norm_1}_n\) and \(\tilde{w}^{norm_2}_n\), \(S\) can simulate them as uniformly random and independent values in \(\mathbb{Z}_{2^k}\) that makes it indistinguishable.

Finally, the inversion computation operation, \(S\) simulates the messages received from other parties during the execution of this function \(\tilde{m}^{inv}_1, \tilde{m}^{inv}_2, \ldots, \tilde{m}^{inv}_n\) and the outcome of the inversion function \(\tilde{w}^{inv}\). According to Theorem \ref{invproof}, the inversion function can be computed privately, and the simulator can successfully simulate these exchanged messages, which are indistinguishable from random values in \(\mathbb{Z}_{2^k}\). Also, since in actual protocol, parties do not reveal the shares of the output of inversion function \(\tilde{w}^{inv}_n\), \(S\) can simulate \(\tilde{w}^{inv}_1\) as uniformly random and independent values in \(\mathbb{Z}_{2^k}\) that makes it indistinguishable.

Now we show the complete simulator view for cosine similarity computation:

\begin{align*}
S(x_I, f_I(x,y)) = (&x_I, y_I, \tilde{m}^{mul_1}_1, \tilde{m}^{mul_1}_2, \ldots, \tilde{m}^{mul_1}_n, \tilde{m}^{mul_2}_1\\
&, \tilde{m}^{mul_2}_2, \ldots, \tilde{m}^{mul_2}_n\, \tilde{m}^{mul_3}_1, \tilde{m}^{mul_3}_2, \ldots,\\
&\tilde{m}^{mul_3}_n\,\tilde{m}^{norm_1}_1, \tilde{m}^{norm_1}_2, \ldots, \tilde{m}^{norm_1}_n,\\
&\tilde{m}^{norm_2}_1, \tilde{m}^{norm_2}_2, \ldots, \tilde{m}^{norm_2}_n, \tilde{m}^{inv}_1, \tilde{m}^{inv}_2\\
&, \ldots, \tilde{m}^{inv}_n)
\end{align*}

Every component of \( S \) is independent and uniformly random within \(\mathbb{Z}_{2^k}\) which represents \( P \)'s private input. As \( S \) is designed to select values that are not only independently and uniformly random but also drawn from the same distribution as those in \(\operatorname{view}_I^\pi\), these components become indistinguishable. This indistinguishability, applicable to all components ensures that they resemble random values, thus reinforcing the privacy aspect of the protocol.

\end{proof}

\begin{theorem}[privacy w.r.t semi-honest behavior]
\label{l2norm}
Consider a secure multi-party computation protocol \(\pi\) designed to compute a function \(f(x, y)\) utilizing the L2-Normalization computation between two vectors (detailed in Section \ref{ComputeL2Normalization}). This protocol involves \(n\) parties, denoted as \(P_1, P_2, \ldots, P_n\), with a subset \(I \subset [n]\) consisting of semi-honest parties where at least one party is honest, while the rest are semi-honest. The protocol \(\pi\) is considered to privately compute the cosine similarity function \(f(x, y)\) if a probabilistic polynomial-time simulator \(S_i\) can be constructed such that, given all potential inputs, the simulated view of \(P_i\) is statistically indistinguishable from \(P_i\)'s actual view during the execution of the protocol.
\end{theorem}

\begin{proof}

The protocol \(\pi\) ensures deterministic correctness by summing the outputs of \(n\) parties modulo \(\mathbb{Z}_{2^k}\). Each party \(P_i\) contributes a specific output \(w_i\) to the function \(f(x, y)\). The privacy of \(\pi\) is shown by proving that for any inputs \(x\) and \(y\), a simulator can generate a view for party \(P\) that is indistinguishable from its actual view during the protocol's execution, encompassing inputs, internal randomness, and received messages.

In the first step of the L2-Normalization function, the norm of each input vector is computed. Then, the inverse of \([\![\mathbf{w}_a]\!]\), the vector to be normalized, is computed. Following this, the inverse value is multiplied by the norm of the reference vector \([\![\mathbf{w}_b]\!]\) using the Beaver multiplication function. Finally, the output of this multiplication is multiplied by \([\![\mathbf{w}_a]\!]\) using the Beaver multiplication function to produce the final output.

We now define the actual view of \(P\) within protocol \(\pi\) for the square function:

\begin{align*}
\operatorname{view}_I^\pi(x,y) = (&x_I, y_I  m^{mul_1}_1,m^{mul_1}_2, \ldots, m^{mul_1}_n, m^{norm_1}_1,\\
&m^{norm_1}_2, \ldots, m^{norm_1}_n, m^{norm_2}_1,m^{norm_2}_2, \ldots,\\
&m^{norm_2}_n, m^{mul_2}_1,m^{mul_2}_2, \ldots, m^{mul_2}_n,m^{inv}_1,\\
&m^{inv}_2, \ldots, m^{inv}_n)
\end{align*}

where $x_I$ is a share of the input of norm computation function, \(m^{mul_1}_1,m^{mul_1}_2, \ldots, m^{mul_1}_n\), and \(m^{mul_2}_1,m^{mul_2}_2, \ldots, m^{mul_2}_n\) are messages received by \(P\) from other parties for the Beaver multiplication operations and \(m^{norm_1}_1, m^{norm_1}_2, \ldots, m^{norm_1}_n\) and \(m^{norm_2}_1, m^{norm_2}_2, \ldots, m^{norm_2}_n\) are messages received by \(P\) from other parties for the norm computation operation and finally \(m^{inv}_1, m^{inv}_2, \ldots, m^{inv}_n\) are messages received by \(P\) from other parties for the inversion operation.

All these elements, are additively shared between the parties or results of the addition of two additively shared values. Therefore, based on the definition in Algorithm \ref{ss}, they are independent uniformly random within the \(\mathbb{Z}_{2^k}\) set. Furthermore, from the $\operatorname{view}^\pi(x,y)$, \(P\) can compute all other relevant quantities and this view contains all of \(P\)'s information. \\

Now, we construct a simulator, denoted \(S\), for the view of party \(P\) that simulates the computation of the cosine similarity computation function. The simulator \(S\) begins by simulating the initial setup, including the of the secret-shared inputs $\tilde{w_a}$ and $\tilde{w_b}$. For the Beaver multiplications, \(S\) simulates the Beaver triples and the messages exchanged during the Beaver multiplication process \(\tilde{m}^{mul_1}_1, \tilde{m}^{mul_1}_2, \ldots, \tilde{m}^{mul_1}_n\) and \(\tilde{m}^{mul_2}_1, \tilde{m}^{mul_2}_2, \ldots, \tilde{m}^{mul_2}_n\) . In simulating the Beaver multiplication function, \(S\) adheres to the privacy guarantees provided by Theorem \ref{beavmulti}. Therefore, the simulator can successfully simulate these exchanged messages, which are indistinguishable from random values in \(\mathbb{Z}_{2^k}\). Also, since in actual protocol, parties do not reveal the shares of the output of Beaver multiplication functions in any steps \(\tilde{w}^{mul_1}_n\),\(\tilde{w}^{mul_2}_n\) and \(S\) can simulate them as uniformly random and independent values in \(\mathbb{Z}_{2^k}\) that makes it indistinguishable.

Next step involves vector norm computations, \(S\) simulates the messages exchanged during the computation process \(\tilde{m}^{norm_1}_1, \tilde{m}^{norm_1}_2, \ldots, \tilde{m}^{norm_1}_n\) and \(\tilde{m}^{norm_2}_1, \tilde{m}^{norm_2}_2, \ldots, \tilde{m}^{norm_2}_n\) . In simulating the norm computation function, \(S\) adheres to the privacy guarantees provided by Theorem \ref{norm}. Therefore, the simulator can successfully simulate these exchanged messages, which are indistinguishable from random values in \(\mathbb{Z}_{2^k}\). Also, since in actual protocol, parties do not reveal the shares of the output of norm functions in any steps \(\tilde{w}^{norm_!}_n\) and \(\tilde{w}^{norm_2}_n\), \(S\) can simulate them as uniformly random and independent values in \(\mathbb{Z}_{2^k}\) that makes it indistinguishable.

Finally, the inversion computation operation, \(S\) simulates the messages received from other parties during the execution of this function \(\tilde{m}^{inv}_1, \tilde{m}^{inv}_2, \ldots, \tilde{m}^{inv}_n\) and the outcome of the inversion function \(\tilde{w}^{inv}\). According to Theorem \ref{invproof}, the inversion  function can be computed privately, and the simulator can successfully simulate these exchanged messages, which are indistinguishable from random values in \(\mathbb{Z}_{2^k}\). Also, since in actual protocol, parties do not reveal the shares of the output of square root function \(\tilde{w}^{inv}_n\), \(S\) can simulate \(\tilde{w}^{inv}_1\) as uniformly random and independent values in \(\mathbb{Z}_{2^k}\) that makes it indistinguishable.

Now we show the complete simulator view for cosine similarity function:

\begin{align*}
S(x_I, f_I(x,y)) = (&x_1, y_1, \tilde{m}^{mul_1}_1, \tilde{m}^{mul_1}_2, \ldots, \tilde{m}^{mul_1}_n, \tilde{m}^{mul_2}_1\\
&, \tilde{m}^{mul_2}_2, \ldots, \tilde{m}^{mul_2}_n\, \tilde{m}^{norm_1}_1, \tilde{m}^{norm_1}_2, \ldots,\\
& \tilde{m}^{norm_1}_n, \tilde{m}^{norm_2}_1, \tilde{m}^{norm_2}_2, \ldots, \tilde{m}^{norm_2}_n,\\
& \tilde{m}^{inv}_1, \tilde{m}^{inv}_2, \ldots, \tilde{m}^{inv}_n)
\end{align*}

Every component of \( S \) is independent and uniformly random within \(\mathbb{Z}_{2^k}\) which represents \( P \)'s private input. As \( S \) is designed to select values that are not only independently and uniformly random but also drawn from the same distribution as those in \(\operatorname{view}_I^\pi\), these components become indistinguishable. This indistinguishability, applicable to all components ensures that they resemble random values, thus reinforcing the privacy aspect of the protocol.

\end{proof}



\end{document}